\documentclass[alpha-refs]{wiley-article}

\usepackage{color}
\usepackage{ulem}
\usepackage{float}

\newtheorem{assumption}{\bf Assumption}

\usepackage{siunitx}

\papertype{Original Article}
\paperfield{Journal Section}

\title{A new model for preferential attachment scheme
with time-varying parameters}

\author[1]{{Bo} {Zhang}}
\author[1]{{Hanyang} {Tian}}
\author[2]{{Guangming} {Pan}}

\affil[1]{University of Science and Technology of China}
\affil[2]{Nanyang Technological University, Singapore}

\corraddress{Bo Zhang, Department of Statistics and Finance, International  Institute  of  Finance, School  of  Management,
University of Science and Technology of China, Hefei, Anhui, 230026, China}
\corremail{wbchpmp@ustc.edu.cn}

\fundinginfo{Chinese National
Natural Science Foundation research grant
No.12001517 \& 72091212; the Double
First-Class Initiative Research Fund YD2040002005;the Fundamental Research Funds
for the Central Universities WK2040000026 \& WK2040000027; MOE Tier 2 Grant 2018-T2-2-112; MOE Tier 1 Grant RG133/18.}

\runningauthor{Zhang et al.}

\begin{document}

\maketitle

\begin{abstract}
We propose an extension of the preferential attachment scheme by allowing the connecting probability to depend on time $t$. 
 We estimate the parameters involved in the model by minimizing the expected squared difference between the number of vertices of degree one and its conditional expectation. The asymptotic properties of the estimators are also investigated when the parameters are time-varying by establishing the central limit theorem (CLT) of the number of vertices of degree one. We propose a new statistic to test whether the parameters have change points. We also offer some methods to estimate the number of change points and detect the locations of change points. Simulations are  conducted to illustrate the performances of the above results.

\end{abstract}

\section{Introduction}

There has been a growing interest in complex networks in recent decades since they can describe a wide variety of real-world systems such as the cell, the internet, airline routes, the social network, etc (See \cite{albert2002statistical} and \cite{leskovec2008statistical}). Understanding the complex social and biological processes that generate these data requires the development of new approaches to formulating and testing network hypotheses. Statistical approaches/models are helpful tools to address
 these questions. Some references that discussed statistical inference about networks, such as \cite{yan2013central}, \cite{bhamidi2018change}, \cite{resnick2016asymptotic}, \cite{ren2016tuning}, and \cite{yan2016asymptotics}.

This paper proposes a model  in the framework of the preferential attachment family(see \cite{barabasi1999emergence}), which is one of the most famous models belonging to evolving networks. Here the evolving network refers to a dynamic network that changes in time. It is an essential subfield of complex networks. Many real-world networks are evolving networks since they are constructed over time.  See the recent surveys \cite{boccaletti2014structure}, \cite{holme2012temporal}, and the references therein.

 The Barab$\acute{a}$si--Albert model has two important ingredients: growth and preferential attachment. Growth means that the nodes and edges in the network increase over time. Preferential attachment means that the more connected a node is, the more likely it will receive new links. In other words, the node with a higher degree has a stronger ability to grab links added to the network. Both growth and preferential attachment exist widely in real networks.

 \subsection{The models and motivation}

We start with the initial Barab$\acute{a}$si--Albert model which is constructed as follows.
\begin{itemize}
\item
Starting with a small number ($m_0$) of nodes, at every time step, add a new node with $m$ ($\leq m_0$) edges that link the new node to $m$ different existing nodes in the network.
\item
When choosing the nodes to which the new node connects, assume the probability $\Pi$ that a new node will be connected to node $i$ is proportional to the degree $d_i$ of node $i$.
\end{itemize}
One can find that there are $(m_0+t)$ nodes and $mt$ edges at any time $t$.  Note that the case $m>1$ is similar to the case $m=1$ for such a model. Hence many papers (e.g. \cite{bollobas2004diameter}, \cite{flaxman2005high} ), only considered the case $m=1$ and extended the results to the case $m>1$ easily. We below consider a case when $m=1$ only.

However, the nodes and edges of the real-world systems often do not satisfy such a relationship (the Barab$\acute{a}$si--Albert model). Instead, the system may gain only new edges but no new nodes. Given this, Chung and Lu (see \cite{chung2006complex}) proposed another type of preferential attachment scheme as follows. Denote by $d(u)$ the degree of the vertex $u$. There is an initial graph $G_0$ having an isolated vertex, and the degree of the vertex is 1. Define two operations on the graph. 1) Vertex-step: add a new vertex $v$, and add an edge $\{u,v\}$ from v by randomly and independently choosing $u$ with probability proportional to  $d(u)$ in the current graph. 2) Edge-step: add an edge $\{u,v\}$ by randomly and independently choosing $u$ and $v$ with probability proportional to the $d(u)$ and $d(v)$ in the current graph.
 Begin with the initial graph $G_0$. For $t>0$, at time t, the graph $G_t$ is formed by modifying the graph $G_{t-1}$ as follows: with probability $p_c$, take a vertex-step. Otherwise, take an edge-step.

 One can find that the model reduces to the Barab$\acute{a}$si--Albert model when $p_c=1$. When $p_c<1$ both models have $t$ edges at the time $t$. However, the number of the nodes in  Chung and Lu's scheme at the time $t$ is a random variable in contrast with the number of the nodes in the Barab$\acute{a}$si--Albert model. Since the probabilities of connection in the two models depend on the degrees of the nodes and the number of the edges at the time $t$, they have many similar results, such as power-laws and the CLT of the number of leaves. One may refer to Chung and Lu (\cite{chung2006complex}) for more details.

Chung and Lu's scheme assumes that the probability of taking a vertex-step is always a constant, $p_c$, at any time $t$. However, the probability $p_c$ may change over time in real-world systems. Another issue is that the probability of choosing the existing node $u$ may be proportional to $(d(u)+a_t)$ instead of $d(u)$ alone, where $a_t$ is an unknown parameter. It means that such a probability at time $t$ may depend on the nodes' degrees, the number of the edges and nodes.  

Given the above considerations, we propose a new preferential attachment model. There is an initial graph $G_0$ having an isolated vertex, and the degree of the vertex is one. Define two operations on the graph.
\begin{itemize}
\item
 Vertex-step: add a new vertex $v$, and add an edge $\{u,v\}$ from $v$ by randomly and independently choosing $u$ with probability proportional to  $(d(u)+a_t)$ in the current graph. Here $a_t>-1$ is a bounded number that depends on the time $t$.
\item
Edge-step: add an edge $\{u,v\}$ by randomly and independently choosing $u$ and $v$ with probability proportional to  $(d(u)+a_t)$ and $(d(v)+a_t)$ in the current graph.
\end{itemize}
 Begin with the initial graph $G_0$.
For $t>0$, $G_t$ is the graph at time $t$ by modifying $G_{t-1}$ by either taking a vertex-step or an edge-step.

Let $v_t$ be the number of vertices in the graph $G_t$. We can find that $G_t$ is formed from $G_{t-1}$ by taking a vertex-step when $v_t-v_{t-1}=1$, by taking an edge-step when $v_t-v_{t-1}=0$.
We make two claims about $v_t$:
\begin{assumption}\label{assumptionA1}
$v_t$ is known (or observed) for any $t$.
\end{assumption}
\begin{assumption}\label{assumptionA2}
There exists a positive constant $p$ such that $\frac{v_t}{t} \geq p$ for any $t>0$.
\end{assumption}

Assumption \ref{assumptionA1} automatically holds for the Barab$\acute{a}$si--Albert model (i.e. \cite{bhamidi2018change}). Moreover, we note that the timestamp for each edge is recorded in many real networks(i.e. Wikipedia talk or Twitter users), and it helps us get the value of  $(v_t-v_{t-1})$ for any time $t$.
 Here we would like to point out that we do not claim whether the probability of taking an edge step or a vertex step is independent of time $t$ or dependent on time $t$. Instead, Assumption  \ref{assumptionA2} is imposed to ensure that $v_t$ has at least an order of $t$.

This paper aims to propose a statistic to estimate $a_t$ involved in the connection probability and  develop a statistic to test whether $a_t$ changes over time. Here change point means that the connection probability, which is the probability of choosing an existing vertex $u$ to connect the newly added vertex $\nu$, changes over time. To this end, we first develop the central limit theorem of the leaves. Here a leaf is a vertex with degree one. Subsequently, the statistics are proposed to estimate at and detect change points. In the Barab$\acute{a}$si--Albert model, the number of edges and nodes are known for each time $t$. \cite{resnick2016asymptotic} and \cite{bhamidi2018change} gave some valuable results in this case. However, in both  Chung and Lu's scheme and our preference attachment scheme, the number of the nodes at the time $t$ is a random variable, and hence the case is much more complicated.

\subsection{Organization of the paper}
The paper is organized as follows. Let $x_t$ be the number of vertices with degree one (leaves) in $G_t$. The CLT for leaves $x_t$ and the estimator of $a_t$ without change point are stated in Section 2.1. We give the estimators of $a_t$ in the different intervals with change points in Section 2.2. Section 3 introduces the methods to test and detect change points of $a_t$. Section 4 is to conduct some simulations to demonstrate the performance of the statistics in Sections 2--3. The proofs of the results are in Appendix.

\section{Main results}

\subsection{Estimator of $a_t$ in an interval without change point}

We start with characterizing the magnitudes of the expectation and variance of the leaves $x_t$, the number of vertices with degree one.
\begin{lemma}\label{t1}
Suppose that Assumption ~\ref{assumptionA2} holds. Then
\begin{equation}\label{1.1m}
0< \liminf\limits_{t \rightarrow \infty} \frac{Ex_t}{t} \leq  \limsup\limits_{t \rightarrow \infty} \frac{Ex_t}{t} \leq 1
\end{equation}
and
\begin{equation}\label{1.1v}
0<\liminf\limits_{t \rightarrow \infty} \frac{Var(x_t)}{t} \leq  \limsup\limits_{t \rightarrow \infty} \frac{Var(x_t)}{t} \leq C
\end{equation}
 for a constant C.
\end{lemma}
We next establish CLT of degree one vertices $x_t$ by using its conditional expectation and conditional variance as normalizing constants. This CLT is crucial to developing the properties of the estimator of change points to be discussed later.  Define the $\sigma$-field $\Gamma_t=\sigma(x_0,\cdots,x_t)$.
\begin{theorem}\label{t2}
 Suppose that Assumptions \ref{assumptionA1}--\ref{assumptionA2} hold.
 Set $0\leq t_0 < t_1< \cdots\ <t_k=T$.

 If $\min\limits_{1 \leq i \leq k}\lim\limits_{T \rightarrow \infty}\frac{t_i-t_{i-1}}{T}>0$, then
 $$\Big(\frac{x_{t_1}-E(x_{t_1} \mid \Gamma_{t_0})}{[Var(x_{t_1} \mid \Gamma_{t_0})]^{1/2}}, \cdots\,,\frac{x_{t_k}-E(x_{t_k} \mid \Gamma_{t_{k-1}})}{[Var(x_{t_k} \mid \Gamma_{t_{k-1}})]^{1/2}}\Big)$$ converges to $N(0, I_k)$ in distribution.
\end{theorem}
Recall that $a_t$ is involved in the connection probability used in either the edge or vertex step.  We are now able to propose an estimator of $a_t$ when there is no change point in some given interval.  To facilitate the presentation, we introduce some notation.
Define two functions (conditional expectation and conditional variance)
\begin{equation}\label{a1}
f_{t_0,t_1}(a)=E(x_{t_1}\mid\Gamma_{t_0}),\quad g_{t_0,t_1}(a)=Var(x_{t_1}\mid\Gamma_{t_0}),
\end{equation}
 where $a_t=a$ for $0 \leq t_0 \leq t \leq t_1$.

We note an important fact that the conditional expectation $E(X|Y)$ is the best predictor of the random variable $X$ in terms of the minimum mean square error criterion when another random variable $Y$ is available. Inspired by it, we propose an estimator of $a_t$ as follows,
\begin{equation}\label{1.01t42gAE}
\hat{a}=\arg\min_{a}(f_{t_0,t_1}(a)-x_{t_1})^2
\end{equation}
when there is no change point in the interval $[t_0,t_1]$. We want to point out that (\ref{1.01t42gAE}) could still yield an estimator by treating $a_t$ as a constant even if it is not in the interval $[t_0,t_1]$. We later prove in Theorems~\ref{t5} and \ref{t6} that the estimators obtained in these two different scenarios behave differently.

We now give a recursive relation between the leaves $x_t$ in Lemma~\ref{l1} below to evaluate $f_{t_0,t_1}(a)$ and $g_{t_0,t_1}(a)$ via numerical analysis schemes. To this end, we first introduce some auxiliary notation.

%


\begin{definition}\label{de1} Let $m_t=\frac{1+a_t}{2t-1+a_tv_{t-1}}$ and $y_t=v_t-v_{t-1}$ where $v_t$ is the number of vertices in the graph. Set
\begin{equation}\label{01.1}
A_t=\left(
  \begin{array}{ccc}
    1-2y_tm_t+(1-y_t)(4m_t^2-4m_t) & (1-y_t)(2m_t-3m_t^2)+y_t(2-m_t) & y_t\\
    0 & (1-m_t)(1-(1-y_t)m_t) & y_t\\
  0 & 0 &1\\
  \end{array}
\right).
\end{equation}
\end{definition}
\begin{lemma}\label{l1}
Let ${\tilde{\mathbf{x}}_t}=(x_t^2,x_t,1)'$,
\begin{equation}\label{1.1}
E({\tilde{\mathbf{x}}_{t+1}} \mid \Gamma_{t})=A_{t+1}{\tilde{\mathbf{x}}_{t}}.
\end{equation}
\end{lemma}
\begin{proposition}
When $a_t=a$ for $0 \leq t_0 \leq t \leq t_1$, $m_t=\frac{1+a}{2t-1+av_{t-1}}$. Hence  (\ref{1.1}) implies that
\begin{equation}\label{1.1bb}
E({\tilde{\mathbf{x}}_{t_1}} \mid \Gamma_{t_0})=\prod_{t=t_0+1}^{t_1} A_{t}{\tilde{\mathbf{x}}_{t_0}}.
\end{equation}
 (\ref{1.1bb}) shows that when $v_t$(or $y_t$) is observed for any $t$, $f_{t_0,t_1}(a)$ and $g_{t_0,t_1}(a)$ only depend on $a$ and $x_{t_0}$. Thus, we can calculate $f_{t_0,t_1}(a)$ and $g_{t_0,t_1}(a)$ with the given $x_{t_0}$ and $v_t$. When $t_0=0$, $x_{t_0}=1$ is non-random. When $t_0>0$, $x_{t_0}$ is random so that  $f_{t_0,t_1}(a)$ and $g_{t_0,t_1}(a)$ are random. However, Theorem~\ref{t2} could provide the convergence of $x_{t_0}$ with $t_0$ tending to infinity by considering the interval $[0,t_0]$. Thus we can get some asymptotic properties of   $f_{t_0,t_1}(a)$ and $g_{t_0,t_1}(a)$ with probability tending to 1.
\end{proposition}

The consistency of the estimate $\hat{a}$ and its fluctuation are built as follows.
\begin{theorem}\label{t5}
 Suppose that Assumptions~\ref{assumptionA1}--\ref{assumptionA2} hold. There are k intervals $B_1, \cdots\,, B_k$ with $B_i=[s_i,t_i]$ satisfing that $0 \leq s_1< t_1<s_2< \cdots\ <s_{k}<t_k=T$.  Assume that $\min\limits_{1 \leq i \leq k}\lim\limits_{T \rightarrow \infty}\frac{t_i-s_i}{T}>0$ and $a_t=\alpha_i$ for any $t \in B_i$. Define the functions $f_{s_i,t_i}(a)$ and $g_{s_i,t_i}(a)$ as in (\ref{a1}). When $T \rightarrow \infty$, there is the unique $\hat{\alpha_i}$ such that $f_{s_i,t_i}(\hat{\alpha}_i)=x_{t_i}$ with probability tending to one.
 Moreover
 $$\Big(\frac{f_{s_1,t_1}'(\hat{\alpha}_1)(\hat{\alpha}_1-\alpha_1)}{g^{1/2}_{s_1,t_1}(\hat{\alpha}_1)}, \cdots\, ,\frac{f_{s_k,t_k}'(\hat{\alpha}_k)(\hat{\alpha}_k-\alpha_k)}{g^{1/2}_{s_k,t_k}(\hat{\alpha}_k)}\Big)$$ converges to $N(0, I_k)$ in distribution.
\end{theorem}

The orders of the functions $f_{t_0,t_1}(a)$ and $g_{t_0,t_1}(a)$ and their respective derivatives are $T$ with probability tending to 1. We specify them in the supplementary material. When there are no change points in the respective intervals (i.e.,$\alpha_t$ is a constant in each subinterval) , Theorem~\ref{t5} states the asymptotic joint distribution for the estimates $\hat{\alpha}_i,i=1,\cdots\,,k$ which are the respective
estimates of $\alpha_t$ in the $k$ subintervals $B_1,\cdots\,,B_k$.

\subsection{The behavior of $\hat{a}$ in the interval with change points}
Theorem~\ref{t5} establishes the central limit theorem for $\hat{a}$ when $a_t$ does not change in an interval. We next investigate the behavior of the estimate $\hat{a}$ when there is a change point in the interval.

The behavior of the estimate $\hat{a}$ obtained from (\ref{1.01t42gAE}) when there is a change point $R$ in an interval $B=[S,T]$ is summarized in the following result.
\begin{theorem}\label{t6}
 Consider an interval $B=[S,T]$ which satisfies that $0 \leq S <R <T$. For any $t \in [S,R) $, $a_t=\alpha_1$, while $a_t=\alpha_2$ for $t \in [R,T] $ with $|\alpha_2-\alpha_1|>0$. Define the function $f_{S,T}(a)$ as before and obtain $\hat{a}$ from (\ref{1.01t42gAE}). Then $\hat{a}$ is unique and $f_{S,T}(\hat{a})=x_{T}$ with probability tending to one when $T \rightarrow \infty$. Moreover
\begin{itemize}
\item[(a)]
if $\lim\limits_{T \rightarrow \infty}\frac{T-R}{T}>0$ and $\lim\limits_{T \rightarrow \infty}\frac{R-S}{T}>0$, then there exists $c>0$ independent of $T$ such that   
$$\min \{|\hat{a}-\alpha_1|, |\hat{a}-\alpha_2| \}>c>0,$$
with probability tending to one when $T \rightarrow \infty$;
\item[(b)]
if $\lim\limits_{T \rightarrow \infty}\frac{T-R}{T}>0$ and $\lim\limits_{T \rightarrow \infty}\frac{R-S}{T}=0$, then  with probability tending to one when $T \rightarrow \infty$
 $$\hat{a} \rightarrow \alpha_2;$$

\item[(c)]
if $\lim\limits_{T \rightarrow \infty}\frac{R-S}{T}>0$ and $\lim\limits_{T \rightarrow \infty}\frac{T-R}{T}=0$, then with probability tending to one when $T \rightarrow \infty$
$$\hat{a} \rightarrow \alpha_1.$$
\end{itemize}
\end{theorem}

In contrast with Theorem~\ref{t5}, Theorem~\ref{t6} indicates that the property of $\hat{a}$ depends on the position of the change point. If the change point is roughly in the middle of the $B=[S,T]$ interval, then $\hat{a}$ deviates from both $\alpha_1$ and $\alpha_2$. However,  $\hat{a}$ converges to the corresponding $\alpha_i$ in probability if the change point is closer to one endpoint of  the $B=[S,T]$ interval.

Based on Theorems~\ref{t5}--\ref{t6}, we can fit a dynamic network by dividing the whole time into some intervals and estimating $a_t$ each interval. For an interval without change points, Theorem \ref{t5} ensures the convergence rate of the estimator.  For an interval with one change point,   Theorem~\ref{t6} ensures the property of the estimator.

\section{Extensions with change points}
\subsection{Testing the existence of the change point of $a_t$}
Theorem~\ref{t5} shows the asymptotic normality of $\hat{a}$. Based on it, we can construct a statistic to test $a_t$. In other words, we can test whether $a_t$ changes over time. Assume that the whole time $[0,T]$ is partitioned into three intervals $[0,R_1]$, $(R_1,R_2]$ and $(R_2,T]$. Here  $R_1$ and $T-R_2$ have the order of $T$, but $R_2-R_1$ can be very small (even one). There is no change point in the first and latest intervals, but  zero or one change point of $a_t$ in $(R_1,R_2]$. Then \begin{equation}\label{changeset}
a_t=\begin{cases}
\alpha_1 & t \leq R_1, \\
\alpha_2 & R_2< t \leq T.
\end{cases}
\end{equation}
Suppose there is no change point in $(R_1,R_2]$, $\alpha_1=\alpha_2$. Otherwise, $\alpha_1 \neq \alpha_2$. Then the direct idea is to use the joint distributions in  Theorem~\ref{t5} to propose the following statistic. Let
 \begin{equation}\label{defstal}
L=\frac{(\hat{\alpha}_{2}-\hat{\alpha}_{1})^2}{\frac{g_{R_2,T}(\hat{\alpha}_{2})}{(f_{R_2,T}'(\hat{\alpha}_{2}))^2}+
 \frac{g_{0,R_1}(\hat{\alpha}_{1})}{(f_{0,R_1}'(\hat{\alpha}_{1}))^2}}.
\end{equation}
Theorem~\ref{t5} implies that the statistic $L$ weakly converges to the chi-square distribution with $1$ degree of freedom when $a_1=a_2$. Moreover, $\frac{g_{R_2,T}(\hat{\alpha}_{2})}{(f_{R_2,T}'(\hat{\alpha}_{2}))^2}+
 \frac{g_{0,R_1}(\hat{\alpha}_{1})}{(f_{0,R_1}'(\hat{\alpha}_{1}))^2}=O_p(T^{-1})$, which implies that $L$ has the order of $(\alpha_1-\alpha_2)^2T$. Thus, $L$ could be used as a statistic to test $\alpha_1=\alpha_2$ or the existence of the change point in $(R_1,R_2]$.

 \subsection{Detecting the interval of the change point for $a_t$}
When we already know there is one change point for $a_t$, we would like to know where the change point is. To this end, define
\begin{equation}\label{defsta2}
\tilde{L}_{[s,t)}=|\hat{\alpha}_{s,1}-\hat{\alpha}_{t,2}|,
\end{equation}
where $\hat{\alpha}_{s,1}$ and $\hat{\alpha}_{t,2}$ are the estimators of $a_t$ based on the time intervals $[0,s)$ and $[t,T]$, respectively.

We detect the (asymptotic) interval containing the change point by the following steps.

(i) Divide the whole time into $k$ intervals $B_1, \cdots, B_k$ where $B_1=[0,s_1)$, $B_i=[s_{i-1},s_{i})$ for $2 \leq i \leq k-1$ and $B_k=[s_{k-1},T)$. There exists $c>0$ such that $s_1>cT$ and $s_{k-1}<(1-c)T$. Let $\hat{j}=\arg\max_{2 \leq i \leq k-1} \tilde{L}_{B_i}$ and define a new interval $V$ dependent on the value of the estimator $\hat{j}$ as follows
\begin{equation*}
V=\begin{cases}
B_{\hat{j}-1}\cup B_{\hat{j}}\cup B_{\hat{j}+1} & 3 \leq \hat{j} \leq k-2, \\
B_{2}\cup B_{3}\cup B_{4} & \hat{j}=2, \\
B_{k-3}\cup B_{k-2}\cup B_{k-1} & \hat{j}=k-1.
\end{cases}
\end{equation*}

(ii) Divide the interval $V$ further into $k$ subintervals $B_1^{1}, \cdots, B_k^{1}$. Let $\hat{j}=\arg\max_{1 \leq i \leq k} \tilde{L}_{B_i^1}$ and define a new interval $V^{(1)}$ as follows:
\begin{equation*}
V^{(1)}=\begin{cases}
B_{\hat{j}-1}^{1}\cup B_{\hat{j}}^{1}\cup B_{\hat{j}+1}^{1} & 2 \leq \hat{j} \leq k-1, \\
B_{1}^{1}\cup B_{2}^{1}\cup B_{3}^{1} & \hat{j}=1, \\
B_{k-2}^{1}\cup B_{k-1}^{1}\cup B_{k}^{1} & \hat{j}=k.
\end{cases}
\end{equation*}

(iii) Repeat Step (ii) until $V^{(q)}$ is small enough. The change point must be in $V^{(q)}$.

\begin{proposition}\label{propchainter}
Let $R$ be the change point such that
 \begin{equation*}
a_t=\begin{cases}
\alpha_1 & 0 \leq t < R, \\
\alpha_2 & R \leq t \leq T.
\end{cases}
\end{equation*}
Moreover, $\|\alpha_1-\alpha_2\|>0$ and $s_1 \leq R \leq s_{k-1}$. Then for any bounded $k$ and $q$,

\begin{equation}\label{defchainter}
\lim_{T \rightarrow \infty}P(R \in V^{(q)})=1.
\end{equation}

\end{proposition}

Proposition \ref{propchainter} can be proved by the following idea. For any finite $d$, if $R \in B_{j}^{d}$, Theorems \ref{t5}-\ref{t6} imply that the maximum of $\tilde{L}_{B_i^d}$ occurs at $i=j-1$, $i=j$ or $i=j+1$ with probability tending to one as $T$ tends to infinity.

The above method can get an interval $V^{(q)}$ containing the change point. The  interval length is not larger than $3^qk^{-q}T$, and we only need to estimate $a_t$ by $2kq$ times. Thus it strikes a balance between the accuracy and the computing cost.

\subsection{Multiple change points}

%
%

Note that $a_t$ may have more than one change point over the whole time. If this is the case, we divide the whole time into more intervals to test the change points. To this end, we first specify an assumption for change points.
\begin{assumption}\label{A3}
The minimal distance between change points are large enough such that the whole time can be partitioned into finite intervals satisfying the following conditions: (i) there is at most one change point in four consecutive intervals and (ii) there is not any change point in the first and last intervals.
\end{assumption}
 Assumption \ref{A3} seems to be restrictive. However, if the distance between each pair of change points is large enough(e.g. $\tau T$ with $\tau>0$), we can partition the whole time into finite intervals with a length not larger than $\tau T/4$. Then Assumption  \ref{A3} holds.



We are now able to propose a statistic to detect the intervals containing change points no matter whether the number of change points is known or unknown. We start with the case when the number of change points is known. Recall $\hat{\alpha}_{i}$ from Theorem \ref{t5}. Given Theorem \ref{t5}, Theorem \ref{t6} and condition Assumption \ref{A3}, we may use the difference between estimators $\hat{\alpha}_{i}$ to locate the positions of change points. The idea behind it is that the asymptotic joint distribution of $(\hat\alpha_i-\hat\alpha_j)$ is Gaussian, which is bounded in probability, if $\alpha_i=\alpha_j$, while $(\hat\alpha_i-\hat\alpha_j)$ diverges if $\alpha_i\neq\alpha_j$ according to the orders of normalizing constants used in Theorem \ref{t5} (see Lemma \ref{ll1} in the supplementary material). Below  we use the difference between $\hat{\alpha}_{i+1}$ and $\hat{\alpha}_{i-1}$ instead of $\hat{\alpha}_{i+1}$ and $\hat{\alpha}_{i}$. It is to avoid being challenging to handle when the change points are near the edge points of the intervals. More precisely,
 define the statistic
 $$L_i=\frac{(\hat{\alpha}_{i+1}-\hat{\alpha}_{i-1})^2}{\frac{g_{s_{i+1},t_{i+1}}(\hat{\alpha}_{i+1})}{(f_{s_{i+1},t_{i+1}}'(\hat{\alpha}_{i+1}))^2}+
 \frac{g_{s_{i-1},t_{i-1}}(\hat{\alpha}_{i-1})}{(f_{s_{i-1},t_{i-1}}'(\hat{\alpha}_{i-1}))^2}}$$ for any $2 \leq i \leq k-1$. 
\begin{theorem}\label{t11}
Suppose that the conditions in Theorem \ref{t5} hold. 
 When $T \rightarrow \infty$, the following results hold with probability tending to one.
\begin{itemize}
\item[(1)] Assume that $a_t=a$ for any $t \in$ $B_{j-1}$, $B_j$ and $B_{j+1}$. Then $L_j$ weakly converges to chi-square distribution with $1$ degree of freedom.
\item[(2)] Suppose that Assumption \ref{A3} holds and there is some $j>1$ and $R \in B_j$ such that
\begin{equation*}
a_t=\begin{cases}
\alpha_1 & t < R, \\
\alpha_2 & t \geq R,
\end{cases}
\end{equation*}
where $|\alpha_1-\alpha_2|>0$. Then there exists $c>0$ such that $L_j>cT$ with probability tending to 1.  Moreover, there is a local maximum in $\{L_{j-k}:k=-1,0,1\}$ with probability tending to 1.
\end{itemize}
\end{theorem}
If the number of change points is $h$, then Theorem \ref{t11} can detect the locations of change points by looking for the first $h$ largest local maximums. Since we do not know the locations of change points when dividing the time into subintervals, sometimes change points may be near to the boundary of two intervals (the cases (b) and (c) of Theorem \ref{t6}). In this case, the local maximum may occur at $(j-1)$ or $(j+1)$. Fortunately, in this case, we can use the method in Section 3.2 to gain the intervals of change points with small lengths. We will give the details in Section 3.4.



We can also detect change points by Theorem \ref{t11} when the number of change points is unknown. To this end, define a critical value $c_T$ that satisfies
$$\lim\limits_{T \rightarrow \infty}c_T=\infty, \quad \lim\limits_{T \rightarrow \infty}\frac{c_T}{T}=0.$$ Denote the number of change points by $s_p$. Its estimator is constructed as
 \begin{equation}\label{31}
 \hat{s}_p=\{\text{the number of the local maximums of $L_i$ bigger than $c_T$}\}.
 \end{equation}
\begin{theorem}\label{t12}
Suppose that the conditions in Theorem \ref{t11} hold. Then when $T \rightarrow \infty$, the following result holds with probability tending one
\begin{itemize}
\item $\hat{s}_p=s_p$.
\end{itemize}
\end{theorem}

However, different $c_T$ may lead to different results when $T$ is not large enough in practice. Thus we propose two different methods to estimate the number of change points when $T$ is not large enough.

The first method is based on the asymptotic distribution of $L_i$. Let $\chi(\beta)$ be the $\beta$ quantile of the chi-square distribution with $1$ degree of freedom. Then we replace $c_T$ by $\chi(0.01/k)$ where $k$ is the number of the intervals (obtained from partitioning the whole time $[0,T]$). If there is no change point in the $(i-1)$th, $i$th and $(i+1)$th intervals, the probability of the event $L_i>\chi(0.01/k)$ tends to $0.01/k$ as $T$ tends to infinity.

The second method does not use (\ref{31}) with a threshold. We sort the local maximums of $L_i$ as $L_{(1)} \geq \cdots \geq L_{(\tilde{k})}$ and define $L_{(i)}/L_{(i+1)}$ as $R_{i}$. Then we let $\hat{s}_p=\arg\max_{1 \leq i \leq \tilde{k}-1}R_{i}$. This method is based on the different orders of $L_i$ from  Theorem \ref{t11}.

\subsection{Two-steps method}
 The above section uses $L_i$ to detect change points. Recall that $k$ means the number of total intervals by partitioning the whole time. When $k$ is small and  $T$ is large enough, Theorems \ref{t11}-\ref{t12} ensure that the method works. However, note that the maximum of $L_i$ in the intervals, which do not have change points, depends on $k$. If $k$ is not small enough and $T$ is not large enough, the methods may not work well. To avoid those problems, we propose a two-step method further:
\begin{itemize}
 \item[(1)] We divide the whole time into $k$ intervals to ensure Assumption \ref{A3}. Note that $k$ should not be too large. Then we can use the methods in Section 3.3 to detect the number of change points such that $\hat{s}_p$ and the first $\hat{s}_p$ largest local maximums of $L_i$ can be obtained.

\item[(2)] If $L_j$ is one of the first $\hat{s}_p$ largest local maximums of $L_i$  and there exist $j_1<j$ and $j_2>j$ such that $L_{k_1}$ and $L_{k_2}$ do not belong to the  $\hat{s}_p$ largest local maximums of $L_i$ for any $j_1\leq k_1<j<k_2 \leq j_2$. Then
 let $\bigcup_{i=j_1+1}^{j_2-1} B_{i}$ be the new whole time interval, and one can use the method in Section 3.2 to obtain the interval containing the change point.
\end{itemize}

\subsection{The locations of change points}
We have proposed the methods to detect intervals containing change points for $a_t$. Sometimes we may want to obtain precise locations of change points rather than intervals. The natural idea is to estimate $a_t$ in rolling windows and find change points in estimators. However, the computation cost is very high since we need to estimate $a_t$ by $T$ times when $T$ is large. Thus, we propose another method to detect the location of the change points:
\begin{itemize}
 \item[(1)] Assume that there is a change point in $(t_0,t_1)$. We estimate $\hat{a}$ by (\ref{1.01t42gAE}). For $t_0 \leq t \leq t_1$, let $\hat{m}_t=\frac{1+\hat{a}}{2t-1+\hat{a}v_{t-1}}$. Let $\hat{x}_{t_0}=x_{t_0}$ and  $\hat{x}_t=(1-\hat{m}_t)[1-(1-y_t)\hat{m}_t]\hat{x}_{t-1}+y_t$ for $t_0+1 \leq t \leq t_1$.

\item[(2)] Let
\begin{equation}\label{cplocdect}
\hat{\tau}=\arg\max_{t \in (t_0,t_1)}|\hat{x}_t-x_t|+1
\end{equation} be the estimator for the location of the change point.
\end{itemize}
\begin{theorem}\label{cplocdect1thm}
Under the conditions of case (a) in Theorem \ref{t6} let
 \begin{equation}\label{cplocdect100a}
h(t)=|E(x_t\mid\Gamma_{S})-f_{S,t}(\hat{a})|
\end{equation}
for $t \in [S,T]$ and
\begin{equation}\label{cplocdect100}
\hat{R}=\arg\max_{t \in (S,T)}h(t).
\end{equation}
Assume that
\begin{equation}\label{cplocdect100b}
h(\hat{R})-h(t)\geq c|t-\hat{R}|
\end{equation}
for some $c>0$ and any $S \leq t \leq T$.
 Let $t_0=S$ and $t_1=T$ in (\ref{cplocdect}). Then
\begin{equation}\label{cplocdect1}
\frac{|\hat{\tau}-1-\hat{R}|}{T}=O_p(\frac{\log T}{T^{1/2}}).
\end{equation}
\end{theorem}
\begin{proposition}\label{proprr}
Under the conditions of case (a) in Theorem \ref{t6} and $y_t=1$ for any $1 \leq t \leq T$, then
\begin{equation}\label{cplocdect1aa}
\lim_{T \rightarrow \infty}P(\hat{R}=R-1)
\end{equation}
and
(\ref{cplocdect100b}) holds.
\end{proposition}
 Simulations in the next section show that (\ref{cplocdect100b}) holds in general cases. One can also combine the two-steps method in the last subsection with the approach in this subsection to detect the location of the change point and hence have a three-step method: At first, we detect the number of change points. Secondly, we obtain the respective intervals containing each change point. Finally, we detect the location of each change point.
\section{Simulation}
This section is to run some simulations to demonstrate the performance of our approach. We design the model as follows:
\begin{itemize}
\item[1.] The whole time is $[0,T]$.
\item[2.] For any $t \in [0,T]$, $y_t=1$ if $\frac{t+1}{2}$ is integer where we define $y_t=v_t-v_{t-1}$.
\item[3.] For any $t \in [\frac{T}{6},\frac{T}{5}]$, $y_{5t}=1$.
\item[4.] Other $y_t=0$.
\item[5.] For any $t \in [1,\frac{23T}{75}]$, $a_t=\alpha_1$.
\item[6.] For any $t \in (\frac{23T}{75},\frac{2T}{3}]$, $a_t=\alpha_2$.
\item[7.] For any $t \in (\frac{2T}{3},T]$, $a_t=\alpha_3$.
\end{itemize}

\subsection{The binary search algorithm}
From the results given in Section 2, we find it important to get the solution of $f_{t_0,t_1}(a)=x_{t_1}$. Since $f_{t_0,t_1}(a)$ is complicated, it is hard to get the solution directly. Fortunately, for any $a$, $f_{t_0,t_1}(a)$ can be calculated by (\ref{1.1}). Moreover, it turns out that $f_{t_0,t_1}(a)$ is monotonously decreasing such that we can use the binary search algorithm to get the solution quickly.
\begin{itemize}
\item Step 1: Initialize three constant $\hat{a}_s$, $\hat{a}_e$ and $e_r$ such that $f_{t_0,t_1}(\hat{a}_e)+e_r < x_{t_1} < f_{t_0,t_1}(\hat{a}_s)-e_r $. In the simulations we set $e_r=0.01$.

\item Step 2: If $f_{t_0,t_1}(\frac{\hat{a}_e+\hat{a}_s}{2})-e_r > x_{t_1}$, let $\hat{a}_s=\frac{\hat{a}_e+\hat{a}_s}{2}$;
If $f_{t_0,t_1}(\frac{\hat{a}_e+\hat{a}_s}{2})+e_r < x_{t_1}$, let $\hat{a}_e=\frac{\hat{a}_e+\hat{a}_s}{2}$.
\item Repeat Step 2 until   $|f_{t_0,t_1}(\frac{\hat{a}_e+\hat{a}_s}{2})-x_{t_1}| \leq e_r $, and then let $\hat{a}=\frac{\hat{a}_e+\hat{a}_s}{2}$.
\end{itemize}

\subsection{The estimator of $a$}
At first, we consider the case when there is no change point in $[0,T]$. Let $\alpha_1=\alpha_2=\alpha_3=1$ and $T=7500$. Divide the whole time into 5 intervals with the same length, then get the estimator of $a_t$ in each interval. The simulation results based on 1000 replications are stated in Table~\ref{table_1}.

\begin{table}[H]
\center
\caption{\bfseries The estimator of $a_t$ when $a_t\equiv1$ and $T=7500$}
\label{table_1}
{
\begin{tabular}{|lccccc|}
\hline
the interval   & (0,T/5]& (T/5,2T/5]& (2T/5,3T/5]&  (3T/5,4T/5]& (4T/5,T]       \\
\hline
mean&1.015 &1.012 &0.991 &1.003 &1.011 \\
variance&0.0326&0.0229 &0.0261&0.0246 &0.0225\\
mean square error&0.0328 &0.0231 &0.0261 &0.0246 &0.0226\\
coverage probability &0.949&0.952&0.947&0.935&0.956\\
\hline
\end{tabular}}
\end{table}

Next, let $T=15000$, and the simulation results based on 1000 replications are reported in Table~\ref{table_2}.
\begin{table}[H]
\center
\caption{\bfseries The estimator of $a_t$ when $a_t\equiv1$ and $T=15000$}
\label{table_2}
{
\begin{tabular}{|lccccc|}

\hline
the interval   & (0,T/5]& (T/5,2T/5]& (2T/5,3T/5]&  (3T/5,4T/5]& (4T/5,T]       \\
\hline
mean&1.010 &1.010 &1.005 &1.007 &1.000 \\
variance&0.0169 &0.0111 &0.0114 &0.0106 &0.0113\\
mean square error&0.0170 &0.0112 &0.0115 &0.0107 &0.0113\\
coverage probability &0.948&0.958&0.950&0.956&0.955\\
\hline
\end{tabular}}
\end{table}

From the above two tables, one can find that when the time $T$ is doubled, the estimator's variance roughly becomes half that of the former one. It matches Theorem~\ref{t5}, which implies that the estimator's variance has the order of $1/T$. Next we consider the normality of $\frac{f_{t_0,t_1}'(\hat{a})(\hat{a}-a)}{g^{1/2}_{t_0,t_1}(\hat{a})}$. When $T=15000$, the normal QQ plots of the five estimators are stated in Figure~\ref{fig:1} .

\begin{figure}[H]
\centering\includegraphics[width=6in,height=4in]{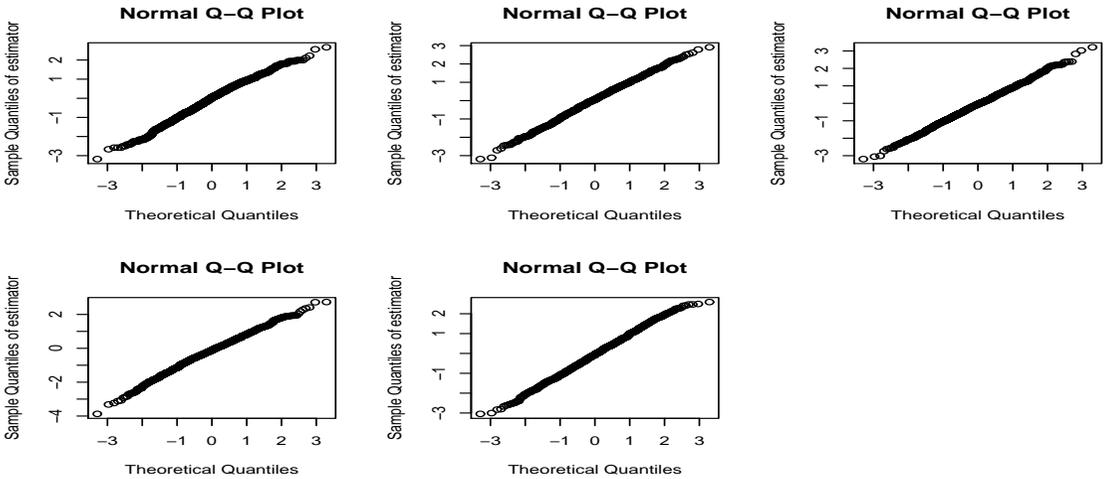}
\caption{the normal QQ plots of the five estimators}\label{fig:1}
\end{figure}

Now, we consider the case when there is a change point on $[0,T]$. Recall $\alpha_t=\alpha_3$ for $t \in (\frac{2T}{3},T]$. Let $\alpha_1=\alpha_2=1$, $\alpha_3=0.5$ and $T=7500$. Divide the whole time into $5$ intervals with the same length. From the table below, the estimator of the change point in the fourth interval can be obtained in the fifth. The simulation results based on $1000$ replications are reported in Table 3. This is consistent with Theorem~\ref{t6}.

\begin{table}[H]
\center
\caption{\bfseries The estimator of $a_t$ when $\alpha_1=\alpha_2=1$ and $\alpha_3=0.5$}
\label{table_3}
{
\begin{tabular}{|lccccc|}
\hline
the interval   & (0,T/5]& (T/5,2T/5]& (2T/5,3T/5]&  (3T/5,4T/5]& (4T/5,T]       \\
\hline
mean&1.007 &1.010 &1.017 &0.644 &0.499  \\
variance&0.0329 &0.0233 &0.0232 &0.0139 &0.0124\\
\hline
\end{tabular}}
\end{table}

\subsection{Testing the change point of $a_t$}

Consider testing the change point by calculating the $L_i$. Let $T=7500$, $a_1=a_2=1,\,a_3\in\{0,1,\cdots\,,6\}.$ Divide the whole time into 5 intervals with the same length, and we can calculate $L_i,\,i\in \{2,3,4\}$ by (\ref{defstal}). For $i=2,3,4$, $(0,R_1]$ is the $(i-1)$th interval and $(R_2,T]$ is the $(i+1)$th interval. $\chi_1(0.05)$ is $L_i$'s  95\% confidence upper limit when there is no change point. Recalling the change point's location and
Theorem \ref{t11}, one can find that $L_2$ follows the chi-square distribution with $1$ degree of freedom with all $a_3$. In contrast, $L_3$ and $L_4$ depend on $a_3$. Table 4 reports the mean of $L_i$ and the proportion of $L_i$ larger than $\chi_1(0.05)$ based on $500$ replications. One can see that both size and power are well.

\begin{table}[H]

\caption{\bfseries the performance of $L_i$, $T=7500,\, a_1=a_2=1$}
\label{table_4}
{
\begin{tabular}{|lccccccc|}
\hline
$a_3$   & 0& 1& 2&  3& 4& 5 &6        \\
\hline
$i=2$&&&&&&&\\
mean of $L_i$&1.0902&1.0401&1.0094&0.8893&0.9400&1.0105& 0.9518 \\
the proportion larger than $\chi_1(0.05)$&0.054&0.042&0.044&0.036&0.04&0.054&0.042\\
\hline
$i=3$&&&&&&&\\
mean of $L_i$&13.5413&1.0099&5.3605&13.1819&18.5253&23.6426& 26.7377  \\
the proportion larger than $\chi_1(0.05)$&0.946&0.052&0.59&0.982&0.994&1&1\\
\hline
$i=4$&&&&&&&\\
mean of $L_i$&28.5188&0.9889&6.9719&13.3156&16.4082&17.4637& 17.1553\\
the proportion larger than $\chi_1(0.05)$&1&0.044&0.792&0.998&1&1&1\\
\hline
\end{tabular}}
\end{table}

\subsection{Detection of the change point in some interval}

Consider the detection of the change point when we know there is only one change point in the interval $[0,T]$. Let  $T=7500,\,\alpha_1=\alpha_2=1$ and $\alpha_3=5$ .  Divide the whole time into $k=5$ intervals with the same length. Moreover, we know the change point is at $t=\frac{2}{3}T=5000$. Detect the location of the change point by the method in Section 3.2 in 500 simulation trials, and record the interval $V^{(q)}$ after $q$th division in each trial in Table \ref{table_5}, where $q\in \{1,2,3\}$.
Table \ref{table_6} shows the proportion of the interval $V^{(q)}$ that contains the change point and the number of different $V^{(q)}$ in the 500 trials.

\begin{table}[H]

\caption{\bfseries The times that $V^{(q)}$ occur in the intervals}
\label{table_5}
{
\begin{tabular}{|r|ccccc|}
\hline
$V^{(q)}$   &&&&&  \\
\hline
$V$& $[1500,6000]:500$&&&&\\
$V^{(1)}$&$[3300,6000]: 500$&&&&\\
$V^{(2)}$&$[4380,6000]:388$&$[3840,5460]: 112$&&&\\
$V^{(3)}$&$[4704,5676]: 166$&$[5028,6000]: 86$&$[4488,5460]: 109$&$[4380,5352]:136$&$[4164,5136]: 3$\\
\hline
\end{tabular}}
\end{table}

\begin{table}[H]

\caption{\bfseries  The proportion of the $V^{(q)}$ contains the change points}
\label{table_6}
{
\begin{tabular}{|r|cccccc|}
\hline
$V^{(q)}$ &$ V $&$V^{(1)}$ &$V^{(2)}$  &$V^{(3)}$ &$V^{(4)}$  &$V^{(5)}$   \\
\hline
Number of the different $V^{(q)}$& 1&1&2&5&12&27\\
Proportion  contains the change points& 100\%& 100\% &100\% & 82.8\% &75.2\% &67.2\%\\
Length of the $V^{(q)}/T$& 0.6 & 0.36 & 0.216 &0.1296 &0.0777 &0.0465 \\
\hline
\end{tabular}}
\end{table}

\subsection{Detection of change points with the number of change points known}
We consider detecting more than one change point with the number of change points known. Let $\alpha_1=\alpha_3=1$ and $\alpha_2=5$. Divide the whole time $[0,T]$ into 30 intervals with the same length. Suppose that there are two change points. The first change point is $\frac{23T}{75}$ in the 10th interval (3T/10,T/3] and another one is $\frac{2T}{3}$ at the boundary between the 20th interval (19T/30,2T/3] and the 21st interval (2T/3,7/10T]. Since the number of change points is two we calculate $L_i$ and seek its first two largest local maximums according to Theorem \ref{t11}. The values of $L_i$ are partly reported based on the simulation results from 100 replications in Table \ref{table_7}. One can see that all the two largest local maximums of $L_i$ occur in the intervals $(4T/15,3T/10]$, $(3T/10,T/3]$, $(19T/30,2T/3]$ and $(2T/3,7/10T]$, which matches Theorem \ref{t11}.

\begin{table}[H]
\caption{\bfseries The times that the two largest local maximums of $L_i$ occur in the intervals}
\label{table_7}
\center
{\begin{tabular}{|r|cccc|}
\hline
$T \setminus intervals$   & (4T/15,3T/10]& (3T/10,T/3]& (19T/30,2T/3]&  (2T/3,7/10T]      \\
\hline
60000&38 &62 &53 &47\\
\hline
\end{tabular}}
\end{table}

\subsection{Detection of change points with the number of change points unknown}

 This section is to use the above model with two change points and the same partition of the whole time. However, we do not know how many change points it has in advance. As before, calculate the values of $L_i$. We need to count how many $L_i$ are bigger than $c_T$ to determine the number of change points according to the paragraph above Theorem \ref{t12}. Note that $c_T$ tends to infinity. We below report only the first four largest local maximums of $L_i$ based on $100$ replications in Table \ref{table_8}. When listing the values of the local maximums of $L_i$ we observe that there is a significant drop from the second-largest local maximum to the third-largest maximum. According to the table below, the average value of the third-largest maximum is $4.237$, and its max value of $11.370$ is small than $\chi(0.01/30)=12.8731$. Therefore we conclude that $\hat{s}_p=2$. After getting the number of change points one can use the same method as Section 3.4 to detect the locations of change points at $\frac{23T}{75}=18400$ in Tables \ref{table_9}-\ref{table_10} and $\frac{2T}{3}=40000$ in Tables \ref{table_11}-\ref{table_12}.

\begin{table}[H]

\center
\caption{\bfseries The first four largest local maximums of $L_i$ in the intervals when $T=60000$}
\label{table_8}
{
\begin{tabular}{|r|cc|cc|}
\hline
 largest local maxmimum   & first& second& third&fourth        \\
\hline
mean&54.012 &43.150 &4.237 &2.602  \\
max&66.930 &50.100 &11.370 &7.578\\
min&43.5458 &34.883 &1.125 &0.766\\
\hline
\end{tabular}}
\end{table}

\begin{table}[H]

\center
\caption{\bfseries The times that $V^{(q)}$ occur in the intervals with change point at $t=18400$}
\label{table_9}
{
\begin{tabular}{|r|cccc|}
\hline
$L$&[16000,22000]: 327 &[14000,20000]: 173&&\\
$V$& $[17200,20800]: 327$& [15200,18800]: 173&&\\
$V^{(1)}$&$[17200,19360]: 136$ & [17920,20080]: 77 & [16640,18800]: 173& [18640,20800]: 114\\
\hline
\end{tabular}}
\end{table}

\begin{table}[H]
\caption{\bfseries  The proportion of the $V^{(q)}$ contains the change point at $t=18400$ }
\label{table_10}
{
\begin{tabular}{|r|cccccc|}
\hline
$V^{(q)}$ &$L$&$ V $&$V^{(1)}$ &$V^{(2)}$  &$V^{(3)}$ &$V^{(4)}$   \\
\hline
Number of the different $V^{(q)}$ & 2&2&4&9&24&58\\
Proportion  contains the change points& 100\%& 100\% &77.2\% & 72.6\% &67.0\% &52.0\%\\
Length of the $V^{(q)}/T$& 0.1& 0.06& 0.036 & 0.0216 &0.0130 &0.0078 \\
\hline
\end{tabular}}
\end{table}

\begin{table}[H]

\caption{\bfseries The times that $V^{(q)}$ occur in the intervals with change point at $t=40000$}
\label{table_11}
{
\begin{tabular}{|r|cccc|}
\hline
$L$&[36000,42000]: 376 &[38000,44000]: 124&&\\
$V$& $[37200,40800]: 376$& [39200,42800]: 124&&\\
$V^{(1)}$&$[38640,40800]: 376$ & [39200,41360]: 110 & [40640,42800]: 8& [39920,42080]: 6\\
\hline
\end{tabular}}
\end{table}

\begin{table}[H]

\caption{\bfseries  The proportion of the $V^{(q)}$ contains the change point at $t=40000$  }
\label{table_12}
{
\begin{tabular}{|r|cccccc|}
\hline
$V^{(q)}$ &$L$&$ V $&$V^{(1)}$ &$V^{(2)}$  &$V^{(3)}$ &$V^{(4)}$   \\
\hline
Number of the different $V^{(q)}$ & 2&2&4&11&22&39\\
Proportion  contains the change points& 100\%& 100\% &98.4\% & 95.6\% &80.6\% &72.4\%\\
Length of the $V^{(q)}/T$& 0.1& 0.06& 0.036 & 0.0216 &0.0130 &0.0078\\
\hline
\end{tabular}}
\end{table}

\subsection{The estimator of $a$ in randomization setting }
This subsection is to consider the $y_t$ in the randomization setting. We design the model as follows:
\begin{itemize}
\item[1.] The whole time is $[0,T]$.
\item[2.] For any $t$, $y_t=1$ with probability $p,\,p\in \{0.1,0.25,0.5,0.75,0.9\}$
\item[3.] For any $t \in [1,\frac{23T}{75}]$, $a_t=\alpha_1$.
\item[4.] For any $t \in (\frac{23T}{75},\frac{2T}{3}]$, $a_t=\alpha_2$.
\item[5.] For any $t \in (\frac{2T}{3},T]$, $a_t=\alpha_3$.
\end{itemize}

At first, consider the case when there is no change point on $[0,T]$. We let $\alpha_1=\alpha_2=\alpha_3=1$ and $T=7500$ the same as in Section 4.2. Then we divide the whole time into 5 intervals with the same length. We can get the estimator of $a_t$ in the five intervals respectively. The simulation results based on 1000 replications are stated in Table~\ref{table_13}.

\begin{table}[H]
\center
\caption{\bfseries The estimator of $a_t$ when $a_t\equiv1$ and $T=7500$}
\label{table_13}
{
\begin{tabular}{|lccccc|}
\hline
the interval   & (0,T/5]& (T/5,2T/5]& (2T/5,3T/5]&  (3T/5,4T/5]& (4T/5,T]       \\
\hline
{$p=0.1$}&&&&&\\
mean     &1.026 &1.008 &1.008 &1.002 &1.014\\
variance &0.0875&0.0525 &0.0515&0.0520 &0.0516\\
mean square error&0.0881 &0.0526 &0.0515 &0.0520 &0.0519\\
coverage probability &0.946&0.955&0.95&0.944&0.947\\
\hline
$p=0.25$&&&&&\\
mean     &1.023 &1.010 &1.008 &1.003 &1.004\\
variance &0.0464&0.0283 &0.0287&0.0263 &0.0232\\
mean square error&0.0469 &0.0284 &0.0288 &0.0263 &0.0233\\
coverage probability &0.949&0.962&0.948&0.945&0.961\\
\hline
$p=0.5$&&&&&\\
mean&1.011 &1.009 &1.003 &1.000 &1.014\\
variance&0.0316&0.0257 &0.0249&0.0228 &0.0236\\
mean square error&0.0317 &0.0258 &0.0250 &0.0228 &0.0238\\
coverage probability &0.952&0.940&0.947&0.95&0.947\\
\hline
$p=0.75$&&&&&\\
mean     &1.019 &0.998 &1.009 &1.007 &1.022\\
variance &0.0341&0.0270 &0.0251&0.0299 &0.0261\\
mean square error&0.0345&0.0270 &0.0252 &0.0299 &0.0266\\
coverage probability &0.949&0.945&0.958&0.943&0.956\\
\hline
$p=0.9$&&&&&\\
mean     &1.0178	 &1.021 &1.015 &1.005 &1.005\\
variance &0.0377&0.0306&0.0338&0.0328 &0.0323\\
mean square error&0.0380 &0.0310 &0.0340 &0.0329 &0.0323\\
coverage probability &0.965 &0.961 &0.948&0.942&0.942\\
\hline
\end{tabular}}
\end{table}

From the above table, one can find that when the probability $p=0.5$, the estimator's variance is smallest and similar to that in Section 4.2 where $y_t$ is known. Although the estimator's performance varies with $p$, the result is still excellent even at $p=0.1$ or $0.9$.

Then consider the case when there is a change point on $[0,T]$. Recall $\alpha_t=\alpha_3$ for $t \in (\frac{2T}{3},T]$. Let $\alpha_1=\alpha_2=1$ and $\alpha_3=0.5$ and $T=7500$. Divide the whole time into $5$ intervals with the same length. The simulation results based on $1000$ replications are reported in Table~\ref{table_14}.
\begin{table}[H]
\center
\caption{\bfseries The estimator of $a_t$ when $\alpha_1=\alpha_2=1$, $\alpha_3=0.5$ and $T=7500$  }
\label{table_14}
{\begin{tabular}{|lccccc|}
\hline
the interval   & (0,T/5]& (T/5,2T/5]& (2T/5,3T/5]&  (3T/5,4T/5]& (4T/5,T]       \\
\hline
$p=0.1$&&&&&\\
mean&1.008 &1.005 &1.014 &0.645 &0.504  \\
variance&0.0850 &0.0532 &0.0545 &0.0373 &0.0325\\
\hline
$p=0.25$&&&&&\\
mean&1.021 &1.008 &1.007 &0.633 &0.503  \\
variance&0.0420 &0.0293 &0.0298 &0.0189 &0.0157\\
\hline
 $p=0.5$&&&&&\\
mean&1.014 &1.005 &1.005 &0.642 &0.506  \\
variance&0.0331 &0.0242 &0.0230 &0.0150 &0.0128\\
\hline
$p=0.75$&&&&&\\
mean&1.011 &1.002 &1.012 &0.639 &0.509  \\
variance&0.0353 &0.0277 &0.0275 &0.0151 &0.0129\\
\hline
$p=0.9$&&&&&\\
mean&1.012 &1.013 &1.000 &0.647 &0.514  \\
variance&0.0407 &0.0344 &0.0323 &0.0183 &0.0167\\
\hline
\end{tabular}}
\end{table}

\subsection{Location of the known change point with $y_t\equiv 1$}
This subsection is to test our method of locating the change point in Section 3.5 without edge-step. We design the model as following:
\begin{itemize}
\item[1.] The whole time is $[0,T]$.
\item[2.] For any $t\in [0,T]$ , $y_t\equiv 1$.
\item[3.] For any $t \in [1,\frac{23T}{75}]$, $a_t=\alpha_1$.
\item[4.] For any $t \in (\frac{23T}{75},\frac{2T}{3}]$, $a_t=\alpha_2$.
\item[5.] For any $t \in (\frac{2T}{3},T]$, $a_t=\alpha_3$.
\end{itemize}

Consider the case that we have known there is only one change point in the interval $[0,T]$.  Let $\alpha_1=\alpha_2=1,\,\alpha_3\in\{0,1,2,3,4,5,6\}$ and $T\in\{7500,15000\}$. The location of the change point is at $R=\frac{2T}{3}+1=0.6667 T+1$. The simulation results based on 1000 replications are stated in Table~\ref{table_15}--\ref{table_16}.

\begin{table}[H]

\caption{\bfseries the performance of $\hat{R}$, $T=7500,\, a_1=a_2=1$}
\label{table_15}
\resizebox{\linewidth}{!}
{
\begin{tabular}{|lcccccc|}
\hline
$a_3$   & 0& 1& 2&  3& 4& 5        \\
\hline
the mean of $\hat{a}$ &0.4185&1.0036&1.4238&1.7603&2.0061&2.2221\\
the mean of $\hat{R}$ &4916.944 &4042.591 & 4730.438&4836.489&4901.25&4920.298\\
the mean of $\frac{|\hat{R}-R|}{T}$ &0.0138&0.2019&0.0485&0.0258&0.0166&0.0134 \\
the mean of $\frac{|\hat{R}-R|^2}{T^2}$ &5.8903e{-04}&0.0615&0.0063&0.0019&9.0794e-04&6.0859e-04\\
the 95\% cover interval of $\frac{\hat{R}}{T}$&[0.5943,0.6801]&[0.1529,0.9076]&[0.4299,0.7259]&[0.5404,0.6960]&[0.5755,0.6825]&[0.5999,0.6801]\\
\hline
\end{tabular}}
\end{table}

\begin{table}[H]

\caption{\bfseries the performance of $\hat{R}$, $T=15000,\, a_1=a_2=1$}
\label{table_16}
\resizebox{\linewidth}{!}
{
\begin{tabular}{|lcccccc|}
\hline
$a_3$   & 0& 1& 2&  3& 4& 5        \\
\hline
the mean of $\hat{a}$ &0.4185&1.0019&1.4300&1.7464&1.9990&2.2061\\
the mean of $\hat{R}$ &9902.4 &8068.795 &9636.639&9826.089&9891.427&9912.922\\
the mean of $\frac{|\hat{R}-R|}{T}$ &0.0080&0.2003&0.0297&0.0141&0.0089&0.0071 \\
the mean of $\frac{|\hat{R}-R|^2}{T^2}$ &2.2765e-04&0.0601&0.0029&7.3896e-04&2.7954e-04&1.7299e-04\\
the 95\% cover interval of $\frac{\hat{R}}{T}$&[0.6199,0.6743]&[0.1495,0.8965]&[0.5055,0.6913]&[0.5841,0.6777]&[0.6122,0.6736]&[0.6239,0.6728]\\
\hline
\end{tabular}}
\end{table}

\subsection{Location of the known change point with randomization setting}
This subsection is to test our method of locating the change point in Section 3.5 when $y_t$ in the randomization setting. We design the model as following:
\begin{itemize}
\item[1.] The whole time is $[0,T]$.
\item[2.] For any $t\in [0,T]$ , $y_t= 1$ with probability $p$.
\item[3.] For any $t \in [1,\frac{23T}{75}]$, $a_t=\alpha_1$.
\item[4.] For any $t \in (\frac{23T}{75},\frac{2T}{3}]$, $a_t=\alpha_2$.
\item[5.] For any $t \in (\frac{2T}{3},T]$, $a_t=\alpha_3$.
\end{itemize}

Consider the case that we have known there is only one change point in the interval $[0,T]$.  Let $\alpha_1=\alpha_2=1,\,\alpha_3=2$, $T\in\{7500,15000\}$ and $p\in\{0.2,0.4,0.6,0.8\}$. The location of the change point is also at $R=\frac{2T}{3}+1=0.6667T+1$. The simulation results based on 1000 replications are stated in Table~\ref{table_16a}--\ref{table_17}.

\begin{table}[H]
\caption{\bfseries the performance of $\hat{R}$, $T=7500,\, a_1=a_2=1,\,a_3=2$}
\label{table_16a}
{
\begin{tabular}{|lcccc|}
\hline
$p$   & 0.2& 0.4& 0.6& 0.8        \\
\hline
the mean of $\hat{a}$ &1.6841&1.5962&1.5386&1.4847\\
the mean of $\hat{R}$ &4722.623 &4783.508 & 4791.303&4743.776\\
the mean of $\frac{|\hat{R}-R|}{T}$ &0.0436&0.0351&0.0347&0.0437 \\
the mean of $\frac{|\hat{R}-R|^2}{T^2}$ &0.0058&0.0039&0.0035&0.0053\\
the 95\% cover interval of $\frac{\hat{R}}{T}$&[0.4271,0.6967]&[0.4814,0.6959]&[0.4845,0.7004]&[0.4642,0.7113]\\
\hline
\end{tabular}}
\end{table}

\begin{table}[H]
\caption{\bfseries the performance of $\hat{R}$, $T=15000,\, a_1=a_2=1,\,a_3=2$}
\label{table_17}
{
\begin{tabular}{|lcccc|}
\hline
$p$   & 0.2& 0.4& 0.6& 0.8        \\
\hline
the mean of $\hat{a}$ &1.6751&1.5983&1.5315&1.4751\\
the mean of $\hat{R}$ &9661.85 &9774.162 & 9740.992&9684.847\\
the mean of $\frac{|\hat{R}-R|}{T}$ &0.0261&0.0182&0.0205&0.0256 \\
the mean of $\frac{|\hat{R}-R|^2}{T^2}$ &0.0022&0.0011&0.0014&0.0019\\
the 95\% cover interval of $\frac{\hat{R}}{T}$&[0.5290,0.6835]&[0.5698,0.6807]&[0.5536,0.6823]&[0.5316,0.6859]\\
\hline
\end{tabular}}
\end{table}

\section*{Acknowledgements}

We are very grateful to Professor Bhamidi and Doctor Jin for discussing with us and kindly providing us with the computer code for the simulations in \cite{bhamidi2018change}.

\section*{appendix: The proof of Theorems}

\subsection*{The proof of the central limit theorem for the number of the leaves $x_t$}

At first we prove Lemma \ref{l1}.
\begin{proof}[Proof of Lemma \ref{l1}]
When $y_t=v_t-v_{t-1}=1$, $G_t$ is formed by modifying $G_{t-1}$ by taking a vertex-step. When $y_t=v_t-v_{t-1}=0$, $G_t$ is formed from modifying $G_{t-1}$ by taking a edge-step. It follows that
\begin{eqnarray}\label{0.1p1}
\nonumber&&P(x_{t+1}=x_t+1\mid\Gamma_{t})=y_t(1-m_tx_t)\\
\nonumber&&P(x_{t+1}=x_t\mid\Gamma_{t})=y_tm_tx_t+(1-y_t)(1-m_tx_t)^2\\
\nonumber&&P(x_{t+1}=x_t-1\mid\Gamma_{t})=(1-y_t)[2m_tx_t(1-m_tx_t)+x_tm_t^2]\\
&&P(x_{t+1}=x_t-2\mid\Gamma_{t})=(1-y_t)(x_t-1)x_tm_t^2.
\end{eqnarray}
These imply that
\begin{equation}\label{1.1p}
E(\mathbf{\tilde{x}_{t+1}}\mid \Gamma_{t})=A_{t+1}\mathbf{\tilde{x}_{t}}.
\end{equation}
So Lemma \ref{l1} follows.
\end{proof}

Now we need some lemmas to prove the central limit theorem for the number of the leaves $x_t$.
\begin{lemma}\label{l2}
Define
 $$z_{i+1}=\frac{\prod\limits_{j=i+2}^{t}  (1-m_j)[1-(1-y_j)m_j][x_{i+1}-(1-m_{i+1})(1-(1-y_{i+1})m_{i+1})x_{i}-y_{i+1}]}{\prod\limits_{j=1}^{t}  (1-m_j)[1-(1-y_j)m_j]}$$  for $0 \leq i \leq t-2$
and $$z_t=\frac{x_{t}-(1-m_{t})[1-(1-y_{t})m_{t}]x_{t-1}-y_{t}}{\prod\limits_{j=1}^{t}  (1-m_j)[1-(1-y_j)m_j]}.$$ Then
\begin{equation}\label{1.1p7}
E(z_i^2\mid\Gamma_{i-1})(\prod\limits_{j=1}^{t}  (1-m_j)^2[1-(1-y_j)m_j]^2) \leq C,
\end{equation}
C is a constant.
\end{lemma}

\begin{proof}[Proof of Lemma \ref{l2}]
Note that for any $i$,
\begin{equation}\label{1.1p3}
E(z_i^2\mid\Gamma_{i-1})= \frac{\prod\limits_{j=i+1}^{t} (1-m_j)^2[1-(1-y_j)m_j]^2 E[(x_{i}-(1-m_{i})[1-(1-y_{i})m_{i}]x_{i-1}-y_{i})^2|\Gamma_{i-1}]}{\prod\limits_{j=1}^{t}  (1-m_j)^2[1-(1-y_j)m_j]^2}.
\end{equation}
It follows from Lemma \ref{l1} that
\begin{eqnarray}\label{1.1p4}
\nonumber &&E[(x_{i}-(1-m_{i})[1-(1-y_{i})m_{i}]x_{i-1}-y_{i})^2\mid\Gamma_{i-1}]=(2-y_i)(1-m_i x_{i-1})m_i x_{i-1}\\
\nonumber&&+[2(1-y_i)(2-y_i)m_i-(1-y_i)^2m_i^2]m_i^2 x_{i-1}^2+3m_i^2(y_i-1)x_{i-1}.
\end{eqnarray}
Note that $m_i \leq 1$ and
\begin{equation}\label{1.1p5}
  \frac{1+a_i}{2+a_i} \frac{x_{i-1}}{i}  \leq m_i x_{i-1}=\frac{1+a_i}{2i-1+a_iv_{i-1}}x_{i-1} \leq \frac{v_{i-1}+a_i v_{i-1}}{2i-1+a_iv_{i-1}}<\frac{1+a_i}{2+a_i}.
\end{equation}
This implies that
\begin{equation}\label{1.1p6}
  E[(x_{i}-(1-m_{i})[1-(1-y_{i})m_{i}]x_{i-1}-y_{i})^2\mid\Gamma_{i-1}]\leq \frac{2-y_i}{4}+7(1-y_i).
\end{equation}
This, together with (\ref{1.1p3}), implies that there exists a constant $C$ such that
\begin{equation*}
(\prod\limits_{j=1}^{t}  (1-m_j)^2[1-(1-y_j)m_j]^2)E(z_i^2\mid\Gamma_{i-1}) \leq C.
\end{equation*}
\end{proof}

\begin{proof}[Proof of Lemma \ref{t1}]

From Lemma \ref{l1} one may obtain
\begin{equation}\label{1.1p1}
Ex_t=y_t+\sum\limits_{i=1}^{t-1} y_i \prod\limits_{j=i+1}^{t}  (1-m_j)[1-(1-y_j)m_j]+x_0\prod\limits_{j=1}^{t}  (1-m_j)[1-(1-y_j)m_j].
\end{equation}
It follows that
\begin{equation}\label{1.1p2}
Ex_t \geq y_t+\sum\limits_{i=[\frac{pt}{2}]}^{t-1} y_i \prod\limits_{j=i+1}^{t}  (1-m_j)[1-(1-y_j)m_j] \geq \prod\limits_{j=[\frac{pt}{2}]+1}^{t}  (1-m_j)[1-(1-y_j)m_j] \sum\limits_{i=[\frac{pt}{2}]}^{t} y_i.
\end{equation}
Recalling Assumption \ref{assumptionA2} we have
$$\sum\limits_{i=[\frac{pt}{2}]}^{t} y_i>\frac{pt}{2}.$$ We can also get a constant $c_0>0$ such that
$$\prod\limits_{j=[\frac{pt}{2}]+1}^{t}  (1-m_j)[1-(1-y_j)m_j]>c_0.$$ Note that $\frac{x_t}{t} \leq \frac{v_t}{t} \leq \frac{t+1}{t}$.  These ensure (\ref{1.1m}) in the main paper.

 Consider (\ref{1.1v}) in the main paper next. Write $$x_t-Ex_t=\Big(\sum\limits_{i=1}^{t}z_i\Big)\Big(\prod\limits_{j=1}^{t}  (1-m_j)[1-(1-y_j)m_j]\Big).$$ From Lemma \ref{l2} one can verify that $\{z_i\}_{1 \leq i \leq t}$ are martingale difference.  A direct calculation yields that
 $$Var(x_t)=\Big(\sum\limits_{i=1}^{t} Var(z_i)\Big)\Big(\prod\limits_{j=1}^{t}  (1-m_j)^2[1-(1-y_j)m_j]^2\Big).$$ Note that $Ez_i=0$ and $Var(z_i)=E(z_i^2)$. This, together with Lemma \ref{l2}, implies that  $Var(x_t)=\sum\limits_{i=1}^{t} Var(z_i) \leq Ct$.

When $t$ is big enough there is $\tilde{t}$ such that $\frac{\tilde{t}}{t} \leq c<1$, $\prod\limits_{j=\tilde{t}+1}^{t} (1-m_j)^2[1-(1-y_j)m_j]^2>c_1$ for a positive number $c_1$ and for $i \geq \tilde{t}$,
\begin{eqnarray}\label{1.1p8}
\nonumber &&E[(x_{i}-(1-m_{i})[1-(1-y_{i})m_{i}]x_{i-1}-y_{i})^2|\Gamma_{i-1}] \\
\nonumber\geq&& \frac{(2-y_i)(1-m_i x_{i-1})m_i x_{i-1}}{2}\\
\geq&& \frac{1}{2(2+a_i)^2}\frac{x_{i-1}}{i}.
\end{eqnarray}
This, together with (\ref{1.1m}) in the main paper and (\ref{1.1p3}), implies that
\begin{eqnarray}\label{1.1p9}
\nonumber &&Var(x_t)\\
\nonumber=&&\sum\limits_{i=1}^{t} Var(z_i)\Big(\prod\limits_{j=1}^{t}  (1-m_j)^2[1-(1-y_j)m_j]^2\Big) \\
\nonumber\geq&& \sum\limits_{i=\tilde{t}}^{t} Var(z_i)\Big(\prod\limits_{j=1}^{t}  (1-m_j)^2[1-(1-y_j)m_j]^2\Big) \\
\geq&& c_1(1-c)t\liminf\limits_{t \rightarrow \infty} \frac{Ex_t}{t} \min\limits_{i}(\frac{1}{2(2+a_i)^2}).
\end{eqnarray}
This leads to (\ref{1.1v}) in the main paper.
\end{proof}

\begin{proof}[Proof of Theorem \ref{t2}]
Recalling the definition of $z_i$ in Lemma \ref{l2} we find $$x_t-E(x_t\mid\Gamma_{t_0})=\Big(\sum\limits_{i=t_0+1}^{t}z_i\Big)\Big(\prod\limits_{j=1}^{t}  (1-m_j)[1-(1-y_j)m_j]\Big).$$
Note that $\sigma-feild$ $\{z_i\}_{1 \leq i \leq t}$ are martingale difference. 
%
Write
\begin{eqnarray}\label{1.1t3}
&&\nonumber\Big(\frac{x_{t_1}-E(x_{t_1}\mid\Gamma_{t_0})}{[Var(x_{t_1}\mid\Gamma_{t_0})]^{1/2}}, \cdots ,\frac{x_{t_k}-E(x_{t_k}\mid\Gamma_{t_{k-1}})}{[Var(x_{t_k}\mid\Gamma_{t_{k-1}})]^{1/2}}\Big)'\\
\nonumber=&&\Big(\frac{(\sum\limits_{i=t_0+1}^{t_1}z_i)\Big(\prod\limits_{j=1}^{t_1}  (1-m_j)[1-(1-y_j)m_j]\Big)}{[Var(x_{t_1}\mid\Gamma_{t_0})]^{1/2}}, \cdots ,\frac{(\sum\limits_{i=t_{k-1}+1}^{t_k}z_i)\Big(\prod\limits_{j=1}^{t_k}  (1-m_j)[1-(1-y_j)m_j]\Big)}{[Var(x_{t_k}\mid\Gamma_{t_{k-1}})]^{1/2}}\Big)'\\
\nonumber=&&\Big(\frac{\sum\limits_{i=t_0+1}^{t_1}z_i}{[Var(\sum\limits_{i=t_0+1}^{t_1}z_i)]^{1/2}}, \cdots ,\frac{\sum\limits_{i=t_{k-1}+1}^{t_k}z_i}{[Var(\sum\limits_{i=t_{k-1}+1}^{t_k}z_i)]^{1/2}}\Big)'.
\end{eqnarray}
Theorem \ref{t2} follows from Lemma \ref{l1}, Lemma \ref{l2} and martingale central limit theorem.

\end{proof}

\subsection*{the proof about the estimator $\hat{a}$}
In order to investigate the estimator $\hat{a}$, the next lemma is to describe the orders of the functions $f_{t_0,t_1}(a)$ and $g_{t_0,t_1}(a)$ and their respective derivatives.
\begin{lemma}\label{ll1}
Suppose that Assumptions \ref{assumptionA1}-\ref{assumptionA2} hold, $a_t=a$ for $0 \leq t_0 \leq t \leq t_1$ and that $\lim\limits_{t_1 \rightarrow \infty}\frac{t_0}{t_1}<1$.  Then $f_{t_0,t_1}(a)$ is monotonically decreasing and has the first two derivatives  for $-1 \leq a < \infty$. Furthermore, the following results hold with probability tending to one
\begin{itemize}
\item[(1)] For $0 \leq t_0 \leq t \leq t_1$,
\begin{equation}\label{1.00t40}
0<\liminf\limits_{t \rightarrow \infty}\frac{f_{t_0,t}(a)}{t}<  \limsup\limits_{t \rightarrow \infty}\frac{f_{t_0,t}(a)}{t}< \infty.
\end{equation}
\item[(2)]\begin{equation}\label{1.00t41}
0<\liminf\limits_{t_1 \rightarrow \infty}\frac{|f_{t_0,t_1}'(a)|}{t_1}< \limsup\limits_{t_1 \rightarrow \infty}\frac{|f_{t_0,t_1}'(a)|}{t_1}< \infty
\end{equation}
and
\begin{equation}\label{1.00t42}
0<\liminf\limits_{t_1 \rightarrow \infty}\frac{|f_{t_0,t_1}''(a)|}{t_1}< \limsup\limits_{t_1 \rightarrow \infty}\frac{|f_{t_0,t_1}''(a)|}{t_1}< \infty
\end{equation}
for $-1 \leq a < \infty$.
\item[(3)]\begin{equation}\label{1.00t41g}
0<\liminf\limits_{t_1 \rightarrow \infty}\frac{g_{t_0,t_1}(a)}{t_1} <\limsup\limits_{t_1 \rightarrow \infty}\frac{g_{t_0,t_1}(a)}{t_1}< \infty
\end{equation}
and
\begin{equation}\label{1.01t42g}
0<\liminf\limits_{t_1 \rightarrow \infty}\frac{g_{t_0,t_1}'(a)}{t_1} <\limsup\limits_{t_1 \rightarrow \infty}\frac{g_{t_0,t_1}'(a)}{t_1} < \infty
\end{equation}
for $-1 \leq a < \infty$.

\end{itemize}
\end{lemma}

\begin{proof}[Proof of Lemma \ref{ll1}]
Recall that $m_t=\frac{1+a_t}{2t-1+a_tv_{t-1}}=\frac{1+a}{2t-1+av_{t-1}}$ for $0 \leq t_0 \leq t \leq t_1$. One can verify that $m_t$ is an increasing and twice differentiable function at $a$ for $-1 \leq a < \infty$. From Lemma \ref{l1} one can find that $f_{t_0,t_1}(a)=E(x_{t_1}\mid\Gamma_{t_0})$ is a decreasing and twice differentiable function at $a$. It follows that $f_{t_0,t_1}(a)$ is a decreasing and twice differentiable function at $a$ for $-1 \leq a < \infty$.

We now prove (\ref{1.00t40}). Recalling Lemma \ref{l1} we can find
\begin{equation*}
|E(x_{t}\mid x_{t_0})-Ex_{t}| \leq |Ex_{t_0}-x_{t_0}|.
\end{equation*}
From Lemma \ref{t1} and Theorem \ref{t2}, $|Ex_{t_0}-x_{t_0}|=O_p(t_0^{1/2})$. This, together with Lemma \ref{t1}, implies (\ref{1.00t40}).

We next consider (\ref{1.00t41})-(\ref{1.01t42g}). Since $f_{t_0,t_1}(a)$ is a decreasing function,  $f_{t_0,t_1}'(a) \leq 0$. Note that $f_{t_0,t_0}(a)=x_{t_0}$ and $f'_{t_0,t_0}(a)=0$. When $t_0< t \leq t_1$
\begin{eqnarray}\label{1.00t41y0}
\nonumber |f_{t_0,t}'(a)|&&=\Big|(1-m_t)[1-(1-y_t)m_t]f_{t_0,t-1}'(a)-f_{t_0,t-1}(a)[1-2(1-y_t)m_t+(1-y_t)]\frac{\partial m_t}{\partial a}\Big|\\
\nonumber  &&\leq |f_{t_0,t-1}'(a)|+ f_{t_0,t-1}(a)\frac{\partial m_t}{\partial a}\\
\nonumber &&\leq |f_{t_0,t-1}'(a)|+t\frac{2t-1-v_{t-1}}{(2t-1+v_{t-1}a)^2}\\
\nonumber&&\leq |f_{t_0,t-1}'(a)|+\frac{t}{(t-1)}\\
\nonumber&&\leq |f_{t_0,t-1}'(a)|+2.
\end{eqnarray}
This, together with $f'_{t_0,t_0}(a)=0$, implies that
\begin{equation}\label{1.00t41y01}
|f_{t_0,t}'(a)|\leq 2(t-t_0).
\end{equation}
This ensures that
\begin{equation}\label{1.00t41r}
\limsup\limits_{t_1 \rightarrow \infty}\frac{|f_{t_0,t_1}'(a)|}{t_1}< \infty.
\end{equation}
Similarly, one can prove that
\begin{equation}\label{1.00t42r}
\limsup\limits_{t_1 \rightarrow \infty}\frac{|f_{t_0,t_1}''(a)|}{t_1}< \infty.
\end{equation}
It follows that
\begin{eqnarray}
\nonumber \frac{|f_{t_0,t}'(a)|}{f_{t_0,t}(a)}&&=\frac{|(1-m_t)[1-(1-y_t)m_t]f_{t_0,t-1}'(a)|}{f_{t_0,t}(a)}\\
\nonumber &&+\frac{f_{t_0,t-1}(a)[1-2(1-y_t)m_t+(1-y_t)]}{f_{t_0,t}(a)}\frac{2t-1-v_{t-1}}{(2t-1+av_{t-1})^2}\\
\nonumber  &&= \frac{|(1-m_t)[1-(1-y_t)m_t]f_{t_0,t-1}'(a)|}{(1-m_t)[1-(1-y_t)m_t]f_{t_0,t-1}(a)+y_t}\\
\nonumber &&+\frac{f_{t_0,t-1}(a)[1-2(1-y_t)m_t+(1-y_t)]}{(1-m_t)[1-(1-y_t)m_t]f_{t_0,t-1}(a)+y_t}\frac{2t-1-v_{t-1}}{(2t-1+av_{t-1})^2}\\
\nonumber&& \geq \frac{|f_{t_0,t-1}'(a)|}{f_{t_0,t-1}(a)}- \frac{|f_{t_0,t-1}'(a)|}{f_{t_0,t-1}(a)}\frac{y_t}{(1-m_t)[1-(1-y_t)m_t]f_{t_0,t-1}(a)+y_t}\\
\nonumber&&+\frac{f_{t_0,t-1}(a)[1-2(1-y_t)m_t+(1-y_t)]}{(1-m_t)[1-(1-y_t)m_t]f_{t_0,t-1}(a)+y_t}\frac{2t-1-v_{t-1}}{(2t-1+av_{t-1})^2}.
\end{eqnarray}
Then
\begin{eqnarray}\label{1.00t41y0aa}
&&\frac{|f_{t_0,t}'(a)|}{f_{t_0,t}(a)}-\frac{|f_{t_0,t-1}'(a)|}{f_{t_0,t-1}(a)}\\
\nonumber &&\geq\frac{1}{(1-m_t)[1-(1-y_t)m_t]f_{t_0,t-1}(a)+y_t}\\
\nonumber&&\{\frac{2t-1-v_{t-1}}{(2t-1+av_{t-1})^2}f_{t_0,t-1}(a)[1-2(1-y_t)m_t+(1-y_t)]-y_t\frac{|f_{t_0,t-1}'(a)|}{f_{t_0,t-1}(a)}\}.
\end{eqnarray}

This, together with (\ref{1.00t40}), implies that there are two positive constant $C_a$ and $D_a$ such that
when $$\frac{|f_{t_0,t-1}'(a)|}{f_{t_0,t-1}(a)}<C_a$$
we have
$$\frac{|f_{t_0,t}'(a)|}{f_{t_0,t}(a)} \geq \frac{|f_{t_0,t-1}'(a)|}{f_{t_0,t-1}(a)}+\frac{D_a}{t}.$$
Since $\lim\limits_{t_1 \rightarrow \infty}\frac{t_0}{t_1}<1$ we have
\begin{equation}\label{1.00t41l}
0<\liminf\limits_{t_1 \rightarrow \infty}\frac{|f_{t_0,t_1}'(a)|}{t_1}.
\end{equation}
This, together with (\ref{1.00t41r}), concludes (\ref{1.00t41}).  (\ref{1.00t42})-(\ref{1.01t42g}) can be proved similarly.

\end{proof}

\begin{lemma}\label{t4}
Suppose that Assumptions \ref{assumptionA1}-\ref{assumptionA2} hold, $a_t=a$ for $0 \leq t_0 \leq t \leq t_1$ and that $\lim\limits_{t_1 \rightarrow \infty}\frac{t_0}{t_1}<1$.  Then the following results hold:
\begin{itemize}
\item $\hat{a}$ is unique and $f_{t_0,t_1}(\hat{a})=x_{t_1}$ with probability tending to one.
\item  $\frac{f_{t_0,t_1}'(\hat{a})(\hat{a}-a)}{g^{1/2}_{t_0,t_1}(\hat{a})}$ converges weakly to standard normal distribution.
  \end{itemize}
\end{lemma}
\begin{proof}[Proof of Lemma \ref{t4}]

We just need to show that there is the unique $\hat{a}$ such that $f_{t_0,t_1}(\hat{a})=x_{t_1}$ with probability tending to one and $\frac{f_{t_0,t_1}'(\hat{a})(\hat{a}-a)}{g^{1/2}_{t_0,t_1}(\hat{a})}$ converges to standard normal distribution.
 Note that $f_{t_0,t_1}(a)$ is continuous. We first prove
 $$\lim\limits_{t_1 \rightarrow \infty}P( \exists -1 < c_a < \infty, f_{t_0,t_1}(c_a) \leq x_{t_1} \leq f_{t_0,t_1}(-1) )=1.$$  As we know, $f_{t_0,t_1}(-1)=x_{t_0}+v_{t_1}-v_{t_0} \geq x_{t_1}$. So it suffices to prove
 $$\lim\limits_{t_1 \rightarrow \infty}P( \exists -1 < c_a < \infty, f_{t_0,t_1}(c_a) \leq x_{t_1} )=1.$$ From the definition of $f_{t_0,t_1}(a)$, we can find $c_a$ and $c_1>0$ such that $f_{t_0,t_1}(c_a)<f_{t_0,t_1}(a)-c_1t_1$. From Lemma \ref{t1} and Theorem \ref{t2}, we can find that $$\lim\limits_{t_1 \rightarrow \infty}P( |x_{t_1}-f_{t_0,t_1}(a)| \leq c_1t_1 )=\lim\limits_{t_1 \rightarrow \infty}P( |x_{t_1}-E(x_{t_1}\mid\Gamma_{t_0})| \leq c_1t_1 )=1.$$We conclude that $\lim\limits_{t_1 \rightarrow \infty}P( \exists -1 < c_a < \infty, f_{t_0,t_1}(c_a) \leq x_{t_1} )=1$. Due to the continuity of $f_{t_0,t_1}(a)$ and $f_{t_0,t_1}'(a)<0$ with probability tending to one we conclude that there is the unique $\hat{a}$ such that $f_{t_0,t_1}(\hat{a})=x_{t_1}$ with probability tending to one.

Recalling Theorem \ref{t2}, $\frac{f_{t_0,t_1}(\hat{a})-f_{t_0,t_1}(a)}{g^{1/2}_{t_0,t_1}(a)}=\frac{x_t-E(x_t\mid\Gamma_{t_0})}{[Var(x_t\mid\Gamma_{t_0})]^{1/2}}$ weakly converges to standard normal distribution. Write
\begin{equation}\label{1.2t4a}
\frac{f_{t_0,t_1}(\hat{a})-f_{t_0,t_1}(a)}{g^{1/2}_{t_0,t_1}(a)}=\frac{f_{t_0,t_1}'(\xi_0)(\hat{a}-a)}{g^{1/2}_{t_0,t_1}(a)},
\end{equation}
where $\xi_0$ is between $a$ and $\hat{a}$.
This, together with (\ref{1.00t41}) and (\ref{1.00t41g}), ensures that
\begin{equation}\label{1.2t4ap}
\hat{a}-a=O_p(t_1^{-1/2}).
\end{equation}
Also write
\begin{equation}\label{1.2t4at}
\frac{f_{t_0,t_1}(\hat{a})-f_{t_0,t_1}(a)}{g^{1/2}_{t_0,t_1}(a)}= \frac{f_{t_0,t_1}'(\hat{a})(\hat{a}-a)-f_{t_0,t_1}''(\xi_1)\frac{(\hat{a}-a)^2}{2}}{g^{1/2}_{t_0,t_1}(\hat{a})}\frac{g^{1/2}(\hat{a})}{g^{1/2}(a)},
\end{equation}
where $\xi_1$ is between $a$ and $\hat{a}$.
From (\ref{1.00t41g}), (\ref{1.01t42g}) and, (\ref{1.2t4ap}), we have $\frac{g(\hat{a})}{g(a)}=1+o_p(1)$. From (\ref{1.00t42}), (\ref{1.00t41g}) and (\ref{1.2t4ap}), it follows that
$$\frac{f_{t_0,t_1}''(\xi_1)\frac{(\hat{a}-a)^2}{2}}{g^{1/2}_{t_0,t_1}(\hat{a})}=o_p(1).$$ We then conclude that $\frac{f_{t_0,t_1}'(\hat{a})(\hat{a}-a)}{g^{1/2}_{t_0,t_1}(\hat{a})}$ weakly converges to standard normal distribution.

\end{proof}
From Theorem \ref{t2} and Lemma \ref{t4}, we can prove  Theorem \ref{t5}.

\begin{proof}[Proof of Theorem \ref{t6}]
We first prove (a). Without loss of generality, we assume that $\alpha_1>\alpha_2$. Then we can get a new function $f_{S,R,T}(\alpha_1,\alpha_2)=E(x_T\mid\Gamma_{S})$. As in the proof of Lemma \ref{t4}, we can find that $\frac{\partial f_{S,R,T}(\alpha_1,\alpha_2)}{\partial \alpha_1} < 0$ and $\frac{\partial f_{S,R,T}(\alpha_1,\alpha_2)}{\partial \alpha_2} < 0$ with probability tending to one. It follows that
$$f_{S,T}(\alpha_1)=f_{S,R,T}(\alpha_1,\alpha_1) < f_{S,R,T}(\alpha_1,\alpha_2) < f_{S,R,T}(\alpha_2,\alpha_2)=f_{S,T}(\alpha_2).$$
Furthermore, we can find that \begin{equation}\label{1.00t61}
0<\liminf\limits_{T \rightarrow \infty}\frac{|\frac{\partial f_{S,R,T}(\alpha_1,\alpha_2)}{\partial \alpha_1}|}{T}< \limsup\limits_{T \rightarrow \infty}\frac{|\frac{\partial f_{S,R,T}(\alpha_1,\alpha_2)}{\partial \alpha_1}|}{T}< \infty
\end{equation}
and
\begin{equation}\label{1.00t62}
0<\liminf\limits_{T \rightarrow \infty}\frac{|\frac{\partial f_{S,R,T}(\alpha_1,\alpha_2)}{\partial \alpha_2}|}{T}< \limsup\limits_{T \rightarrow \infty}\frac{|\frac{\partial f_{S,R,T}(\alpha_1,\alpha_2)}{\partial \alpha_2}|}{T}< \infty.
\end{equation}
These, together with Theorem \ref{t2}, imply that $$f_{S,T}(\alpha_1) < x_T < f_{S,T}(\alpha_2)$$ and
\begin{equation}\label{1.00t63}
\liminf\limits_{T \rightarrow \infty}\frac{\min\{|x_T-f_{S,T}(\alpha_1)|,  |x_T-f_{S,T}(\alpha_2)|\}}{T}>0
\end{equation}
with probability tending to one. So $\hat{a}$ is the unique and $f_{S,T}(\hat{a})=x_{T}$ with probability tending to one. (\ref{1.00t61})-(\ref{1.00t63}) and (\ref{1.00t41}) also imply that there is $c>0$ such that $\min \{|\hat{a}-\alpha_1|, |\hat{a}-\alpha_2| \}>c$ with probability tending to one.

The proof of (b) is similar to the one of (a). Without loss of generality, we assume that $\alpha_1>\alpha_2$. Then we can get a new function $f_{S,R,T}(\alpha_1,\alpha_2)=E(x_T\mid\Gamma_{S})$. As in the proof of Lemma \ref{t4}, we can find that $\frac{\partial f_{S,R,T}(\alpha_1,\alpha_2)}{\partial \alpha_1} < 0$ and $\frac{\partial f_{S,R,T}(\alpha_1,\alpha_2)}{\partial \alpha_2} < 0$ with probability tending to one. It follows that
$$f_{S,T}(\alpha_1)=f_{S,R,T}(\alpha_1,\alpha_1) < f_{S,R,T}(\alpha_1,\alpha_2) < f_{S,R,T}(\alpha_2,\alpha_2)=f_{S,T}(\alpha_2).$$
Furthermore, we can find that
\begin{equation}\label{1.00t72}
0<\liminf\limits_{T \rightarrow \infty}\frac{|\frac{\partial f_{S,R,T}(\alpha_1,\alpha_2)}{\partial \alpha_2}|}{T}< \limsup\limits_{T \rightarrow \infty}\frac{|\frac{\partial f_{S,R,T}(\alpha_1,\alpha_2)}{\partial \alpha_2}|}{T}< \infty.
\end{equation}
Now we note that $\lim\limits_{T \rightarrow \infty}\frac{R-S}{T}=0$ and
\begin{equation}\label{1.00t71}
\lim\limits_{T \rightarrow \infty}\frac{|\frac{\partial f_{S,R,T}(\alpha_1,\alpha_2)}{\partial \alpha_1}|}{T}=0.
\end{equation}
This, together with (\ref{1.00t72}) and Theorem \ref{t2}, implies that $$f_{S,T}(\alpha_1) \leq x_T < f_{S,T}(\alpha_2)$$ and
\begin{equation}\label{1.00t73}
\lim\limits_{T \rightarrow \infty}\frac{|x_T-f_{S,T}(\alpha_2)|}{T}=0
\end{equation}
with high probability. So $\hat{a}$ is the unique and $f_{S,T}(\hat{a})=x_{T}$ with probability tending to one. (\ref{1.00t72})-(\ref{1.00t73}) and (\ref{1.00t41}) also imply that  $\hat{a} \rightarrow \alpha_2$ in probability.

The proof of (c) is similar to the one of (b). We omit the details.
\end{proof}
\subsection*{Proof of Theorems about change points}
\begin{proof}[Proof of Theorem \ref{t11}]
When $a_t=a$ for any $t \in B_{j-1},B_j,B_{j+1}$, Theorem \ref{t2} shows that
 $$\Big(\frac{f_{s_{j-1},t_{j-1}}'(\hat{\alpha}_{j-1})(\hat{\alpha}_{j-1}-a)}{g^{1/2}_{s_{j-1},t_{j-1}}(\hat{\alpha}_{j-1})} ,\frac{f_{s_{j+1},t_{j+1}}'(\hat{\alpha}_{j+1})(\hat{\alpha}_{j+1}-a)}{g^{1/2}_{s_{j+1},t_{j+1}}(\hat{\alpha}_{j+1})}\Big)$$ converges to $N(0, I_2)$ in distribution.
Thus, $L_j$ weakly converges to the chi-square distribution with 1 degree freedom.

Similarly, when $R \in B_j$ is a change point. Assumption \ref{A3} ensures there is no more change point in  $B_{j-1},B_j,B_{j+1}$. Then $a_t=\alpha_1$ for $t \in B_{j-1}$ and $a_t=\alpha_2$ for $t \in B_{j+1}$. Theorem \ref{t2} shows that
 $$\Big(\frac{f_{s_{j-1},t_{j-1}}'(\hat{\alpha}_{j-1})(\hat{\alpha}_{j-1}-\alpha_1)}{g^{1/2}_{s_{j-1},t_{j-1}}(\hat{\alpha}_{j-1})} ,\frac{f_{s_{j+1},t_{j+1}}'(\hat{\alpha}_{j+1})(\hat{\alpha}_{j+1}-\alpha_2)}{g^{1/2}_{s_{j+1},t_{j+1}}(\hat{\alpha}_{j+1})}\Big)$$ converges to $N(0, I_2)$ in distribution.
This, together with $|\alpha_1-\alpha_2|>0$, (\ref{1.00t41}) and (\ref{1.00t41g}), concludes $L_j>cT$ for some $c>0$ with the probability tending to 1.

Moreover, when $R \in B_j$ is a change point, Assumption \ref{A3} ensures there is no more change point in  $\cup_{k=-3}^3B_{j+k}$. Thus $L_{j-2}$ and $L_{j+2}$ weakly converge to the chi-square distribution with 1 degree freedom. This, together with $L_j>cT$ for some $c>0$, completes the proof.

\end{proof}

\begin{proof}[Proof of Theorem \ref{t12}]
The second part of Theorem \ref{t11} shows that each change point produces a local maximum that is larger than $cT$ for some $c>0$ with the probability tending to 1. Thus $\lim_{T \rightarrow \infty}P(\hat{s}_p \geq s_p)=1$.

Similarly, the first part of Theorem \ref{t11} shows that $L_j$ is $O_p(1)=o_p(c_T)$ when $a_t=a$ for any $t \in B_{j-1},B_j,B_{j+1}$. Thus  $\lim_{T \rightarrow \infty}P(\hat{s}_p \leq s_p)=1$. We complete the proof.

\end{proof}
\begin{proof}[Proof of Theorem \ref{cplocdect1thm}]
Without loss of generality, we assume that $\alpha_1>\alpha_2$. Theorem \ref{t6} shows that
$$\min \{\alpha_1-\hat{a}, \hat{a}-\alpha_2 \}>c>0,$$
with probability tending to one when $T \rightarrow \infty$.
$\hat{x}_t=f_{S,t}(\hat{a})$ for $S+1 \leq  t \leq T$.
\begin{eqnarray*}
&&|x_t-\hat{x}_t|=|E(x_t\mid\Gamma_{S})-f_{S,t}(\hat{a})+[x_t-E(x_t\mid\Gamma_{S})]|
\end{eqnarray*}
Then
\begin{eqnarray*}
|E(x_t\mid\Gamma_{S})-f_{S,t}(\hat{a})|-|x_t-E(x_t\mid\Gamma_{S})|  \leq |x_t-\hat{x}_t|\leq |E(x_t\mid\Gamma_{S})-f_{S,t}(\hat{a})|+|x_t-E(x_t\mid\Gamma_{S})|.
\end{eqnarray*}
Using the martingale differences  in Lemma \ref{l2} and (\ref{1.1p6}), we can find that
\begin{eqnarray*}
\max_{t \in [S,T]}|x_t-E(x_t\mid\Gamma_{S})|=O_p(T^{1/2}\log T).
\end{eqnarray*}
These, together with (\ref{cplocdect100b}), complete the proof.

\end{proof}
\begin{proof}[Proof of Proposition \ref{proprr}]
Since $y_t=1$ for any $1 \leq t \leq T$,
\begin{eqnarray}\label{dde}
E(x_{t}\mid\Gamma_{S})=(1-\frac{1+a_t}{2t-1+a_tt})E(x_{t-1}\mid\Gamma_{S})+1,
\end{eqnarray}
for $S+1 \leq t \leq T$. Here $a_t=\alpha_1$ when $S+1 \leq t < R$ and $a_t=\alpha_2$ when $R \leq t \leq T$.
Without loss of generality, we assume that $\alpha_1>\alpha_2$. Theorem \ref{t6} shows that
$$\min \{\alpha_1-\hat{a}, \hat{a}-\alpha_2 \}>c>0,$$
with probability tending to one when $T \rightarrow \infty$.
Thus we only consider the case that
\begin{eqnarray}\label{inpropa1a2a}
\min \{\alpha_1-\hat{a}, \hat{a}-\alpha_2 \}>c>0.
\end{eqnarray}
 Since $E(x_S\mid\Gamma_{S})=x_{S}=f_{S,S}(\hat{a})$ and $\alpha_1-\hat{a}>c$, (\ref{dde}) implies that when $S+1 \leq t < R$,
\begin{eqnarray}\label{dde1}
&&f_{S,t}(\hat{a})-E(x_t\mid\Gamma_{S})>0.
\end{eqnarray}
 On the other hand, Theorem \ref{t2} and (\ref{1.01t42gAE}) imply that
$E(x_T\mid\Gamma_{S})=x_{T}+O_p(T^{1/2})$ and $f_{S,T}(\hat{a})=x_{T}$. Then $E(x_T\mid\Gamma_{S})-f_{S,T}(\hat{a})=O_p(T^{1/2})$.
 This, together with $y_t \equiv 1$, $E(x_S\mid\Gamma_{S})=x_{S}=f_{S,S}(\hat{a})$, (\ref{1.00t41}) and $\hat{a}-\alpha_2>c$, implies that if $T_2 \in (S,T]$ satisfies $E(x_{T_2}\mid\Gamma_{S})=f_{S,T_2}(\hat{a})$, then $T-T_2=O_p(T^{1/2})$. Let $T_3=\min\{S<t\leq T:E(x_{t}\mid\Gamma_{S})=f_{S,t}(\hat{a}) \}$. If $\{S<t\leq T:E(x_{t}\mid\Gamma_{S})=f_{S,t}(\hat{a}) \}$ is empty, let $T_3=T$.
Then we can divide $(S,T]$ into three parts:$[S+1,R)$, $[R,T_3)$ and $[T_3,T]$.

When $S+1 \leq t < R$,
\begin{eqnarray*}
&&f_{S,t}(\hat{a})-E(x_t\mid\Gamma_{S})>0.
\end{eqnarray*}

\begin{eqnarray}\label{rrch22}
&&f_{S,t}(\hat{a})-E(x_t\mid\Gamma_{S})-f_{S,t-1}(\hat{a})+E(x_{t-1}\mid\Gamma_{S})\\
\nonumber=&&\frac{1+\alpha_1}{2t-1+\alpha_1t}E(x_{t-1}\mid\Gamma_{S})-\frac{1+\hat{a}}{2t-1+\hat{a}t}f_{S,t-1}(\hat{a}).
\end{eqnarray}
\begin{eqnarray*}
&&\frac{f_{S,t}(\hat{a})}{E(x_t\mid\Gamma_{S})}-\frac{f_{S,t-1}(\hat{a})}{E(x_{t-1}\mid\Gamma_{S})}\\
=&&\frac{f_{S,t}(\hat{a})-E(x_t\mid\Gamma_{S})}{E(x_t\mid\Gamma_{S})}-\frac{f_{S,t-1}(\hat{a})-E(x_{t-1}\mid\Gamma_{S})}{E(x_{t-1}\mid\Gamma_{S})}\\
=&&\frac{1+\alpha_1}{2t-1+\alpha_1t}-\frac{1+\hat{a}}{2t-1+\hat{a}t}\frac{f_{S,t-1}(\hat{a})}{E(x_{t-1}\mid\Gamma_{S})}\\
+&&\frac{f_{S,t}(\hat{a})-E(x_t\mid\Gamma_{S})}{E(x_t\mid\Gamma_{S})}-\frac{f_{S,t}(\hat{a})-E(x_t\mid\Gamma_{S})}{E(x_{t-1}\mid\Gamma_{S})}\\
=&&\frac{1+\hat{a}}{2t-1+\hat{a}t}[\frac{1+\alpha_1}{1+\hat{a}}\frac{2t-1+\hat{a}t}{2t-1+\alpha_1t}-\frac{f_{S,t-1}(\hat{a})}{E(x_{t-1}\mid\Gamma_{S})}]\\
+&&\frac{f_{S,t}(\hat{a})-E(x_t\mid\Gamma_{S})}{E(x_t\mid\Gamma_{S})}-\frac{f_{S,t}(\hat{a})-E(x_t\mid\Gamma_{S})}{E(x_{t-1}\mid\Gamma_{S})}.
\end{eqnarray*}
Then
\begin{eqnarray*}
&&\frac{f_{S,t}(\hat{a})}{E(x_t\mid\Gamma_{S})}-\frac{f_{S,t-1}(\hat{a})}{E(x_{t-1}\mid\Gamma_{S})}\\
<&&\frac{1}{2}[\frac{1+\alpha_1}{1+\hat{a}}\frac{2t-1+\hat{a}t}{2t-1+\alpha_1t}-\frac{f_{S,t-1}(\hat{a})}{E(x_{t-1}\mid\Gamma_{S})}]
\end{eqnarray*}
for any $t \geq S+2$.
Then for some $S+2 \leq  t <R$,
if \begin{eqnarray}\label{rrch0}
\frac{f_{S,t-1}(\hat{a})}{E(x_{t-1}\mid\Gamma_{S})} < \frac{1+\alpha_1}{1+\hat{a}}\frac{2t-1+\hat{a}t}{2t-1+\alpha_1t},
\end{eqnarray}
we have
 \begin{eqnarray}\label{rrch1}
\frac{f_{S,t}(\hat{a})}{E(x_{t}\mid\Gamma_{S})} < \frac{1+\alpha_1}{1+\hat{a}}\frac{2t+1+\hat{a}(t+1)}{2t+1+\alpha_1(t+1)}.
\end{eqnarray}
Moreover,
\begin{eqnarray*}
&&\frac{f_{S,S+1}(\hat{a})}{E(x_{S+1}\mid\Gamma_{S})}\\
=&&\frac{1+\hat{a}}{2S+1+\hat{a}(S+1)}[\frac{1+\alpha_1}{1+\hat{a}}\frac{2S+1+\hat{a}(S+1)}{2S+1+\alpha_1(S+1)}-\frac{f_{S,S}(\hat{a})}{E(x_{S}\mid\Gamma_{S})}]\\
+&&\frac{f_{S,S+1}(\hat{a})-E(x_{S+1}\mid\Gamma_{S})}{E(x_t\mid\Gamma_{S})}-\frac{f_{S,S+1}(\hat{a})-E(x_{S+1}\mid\Gamma_{S})}{E(x_{S}\mid\Gamma_{S})}+\frac{f_{S,S}(\hat{a})}{E(x_{S}\mid\Gamma_{S})}\\
=&&\frac{1+\hat{a}}{2S+1+\hat{a}(S+1)}[\frac{1+\alpha_1}{1+\hat{a}}\frac{2S+1+\hat{a}(S+1)}{2S+1+\alpha_1(S+1)}-1]\\
+&&\frac{f_{S,S+1}(\hat{a})-E(x_{S+1}\mid\Gamma_{S})}{E(x_t\mid\Gamma_{S})}-\frac{f_{S,S+1}(\hat{a})-E(x_{S+1}\mid\Gamma_{S})}{E(x_{S}\mid\Gamma_{S})}+1\\
<&&\frac{1+\alpha_1}{1+\hat{a}}\frac{2S+1+\hat{a}(S+1)}{2S+1+\alpha_1(S+1)}.
\end{eqnarray*}
This, together with the discussion in (\ref{rrch0})-(\ref{rrch1}), implies that
%
%
%
%
%
\begin{eqnarray}\label{rrch0cc}
\frac{f_{S,t-1}(\hat{a})}{E(x_{t-1}\mid\Gamma_{S})} < \frac{1+\alpha_1}{1+\hat{a}}\frac{2t-1+\hat{a}t}{2t-1+\alpha_1t}
\end{eqnarray}
for $S+1 \leq t <R$.
This, together with (\ref{rrch22}), implies that
\begin{eqnarray}\label{rrch0cc1}
f_{S,t}(\hat{a})-E(x_t\mid\Gamma_{S})>f_{S,t-1}(\hat{a})-E(x_{t-1}\mid\Gamma_{S})
\end{eqnarray}
for $S+1 \leq t <R$.
Now we consider the second part $[R,T_3)$.
When $R \leq t < T_3$,
\begin{eqnarray*}
&&f_{S,t}(\hat{a})>E(x_t\mid\Gamma_{S})
\end{eqnarray*}
and
\begin{eqnarray}\label{rrch22vv}
&&f_{S,t}(\hat{a})-E(x_t\mid\Gamma_{S})-f_{S,t-1}(\hat{a})+E(x_{t-1}\mid\Gamma_{S})\\
\nonumber=&&\frac{1+\alpha_2}{2t-1+\alpha_1t}E(x_{t-1}\mid\Gamma_{S})-\frac{1+\hat{a}}{2t-1+\hat{a}t}f_{S,t-1}(\hat{a}).
\end{eqnarray}
These, together with $\frac{1+\alpha_2}{2t-1+\alpha_1t}<\frac{1+\hat{a}}{2t-1+\hat{a}t}$, implies
\begin{eqnarray}\label{rrch0cc2}
f_{S,t}(\hat{a})-E(x_t\mid\Gamma_{S})<f_{S,t-1}(\hat{a})-E(x_{t-1}\mid\Gamma_{S}).
\end{eqnarray}
For $R \leq t < T_3$.
(\ref{rrch0cc1}) and (\ref{rrch0cc2}) imply that $|f_{S,R-1}(\hat{a})-E(x_{R-1}\mid\Gamma_{S})|>|f_{S,t}(\hat{a})-E(x_t\mid\Gamma_{S})|$ for any $t \in (S,T_2]$ and $t \neq R-1$.
Moreover, (\ref{1.00t41}) implies that $|f_{S,R-1}(\hat{a})-E(x_{R-1}\mid\Gamma_{S})|>c_2T$ for some $c_2>0$. For any $t \in [T_3,T]$, $|f_{S,t}(\hat{a})-E(x_t\mid\Gamma_{S})|=O_p(T^{1/2})$ due to $T-T_3=O_p(T^{1/2})$, $f_{S,T_3}(\hat{a})=E(x_{T_3}\mid\Gamma_{S})$ and $f_{S,T}(\hat{a})=E(x_{T}\mid\Gamma_{S})+O_p(T^{1/2})$. We prove (\ref{cplocdect1aa}).

(\ref{rrch22}) and (\ref{rrch0cc}) imply $$f_{S,t}(\hat{a})-E(x_t\mid\Gamma_{S})-f_{S,t-1}(\hat{a})+E(x_{t-1}\mid\Gamma_{S})>c_3>0,$$
for $t \in [S+1,R)$. (\ref{rrch22vv}), (\ref{1.1m}), Theorem \ref{t2} and (\ref{1.00t40}) imply $$f_{S,t-1}(\hat{a})-E(x_{t-1}\mid\Gamma_{S})-f_{S,t}(\hat{a})+E(x_t\mid\Gamma_{S})>c_3>0,$$
for $t \in [R,T_3)$. These, together with $|f_{S,R-1}(\hat{a})-E(x_{R-1}\mid\Gamma_{S})|>c_2T$ and $|f_{S,t}(\hat{a})-E(x_t\mid\Gamma_{S})|=O_p(T^{1/2})$  for any $t \in [T_3,T]$, complete the proof of (\ref{cplocdect100b}).
\end{proof}



\bibliography{sample}

\end{document}